\documentclass[journal,compsoc]{IEEEtran}
\usepackage{amssymb}
\usepackage{amsmath}
\usepackage{amsthm}
\usepackage{amsfonts}
\usepackage{cite}
\usepackage{graphicx}
\usepackage{epstopdf}
\usepackage{xcolor}
\usepackage{mathrsfs}
\usepackage{enumitem}%
\usepackage{caption}
\usepackage{subcaption}
\usepackage{balance}
\usepackage{url}
\usepackage{times}
\usepackage{bbm}
\usepackage[T1]{fontenc}

\usepackage{miama}
\usepackage[T1]{fontenc}

\providecommand{\U}[1]{\protect\rule{.1in}{.1in}}

\newtheorem{theorem}{Theorem}

\newtheorem{condition}{Assumption}

\newtheorem{corollary}{Corollary}

\newtheorem{definition}{Definition}

\newtheorem{lemma}{Lemma}

\newtheorem{proposition}{Proposition}

\begin{document}

\title{Contact Adaption during Epidemics: A Multilayer Network Formulation Approach\\
\large{(This paper is published in the IEEE Transactions on Network Science and Engineering)}}
\author{Faryad Darabi Sahneh$^{1,*}$, Aram Vajdi$^1$, Joshua Melander$^1$, and~Caterina M.
Scoglio$^1$\thanks{This material is based on work supported by the National
Science Foundation under Grant No. CIF-1423411.}\thanks{$^1$ Department of Electrical and Computer
Engineering, Kansas State University, Manhattan, KS}\thanks{$^*$ Corresponding author: \texttt{faryad@ksu.edu}}}

\IEEEtitleabstractindextext{
\begin{abstract}
People change their physical contacts as a preventive response to infectious disease propagations. Yet, only a few mathematical models consider the coupled dynamics of the disease propagation and the contact adaptation process. This paper presents a model where each agent has a default contact neighborhood set, and switches to a different contact set once she becomes alert about infection among her default contacts. Since each agent can adopt either of two possible neighborhood sets, the overall contact network switches among $2^{N}$ possible configurations. Notably, a two-layer network representation can fully model the underlying adaptive, state-dependent contact network.
Contact adaptation influences the size of the disease prevalence and the epidemic threshold---a characteristic measure of a contact network robustness against epidemics---in a nonlinear fashion. Particularly, the epidemic threshold for the presented adaptive contact network belongs to the solution of a nonlinear Perron-Frobenius (NPF) problem, which does not depend on the contact adaptation rate monotonically. Furthermore, the network adaptation model predicts a counter-intuitive scenario where adaptively changing contacts may adversely lead to lower network robustness against epidemic spreading if the contact adaptation is not fast enough. An original result for a class of NPF problems facilitate the analytical developments in this paper.

\end{abstract}

}

\maketitle

\IEEEdisplaynontitleabstractindextext
\IEEEpeerreviewmaketitle

\IEEEraisesectionheading{\section{Introduction\label{Introduction}}}


\IEEEPARstart{M}{athematical} models of infectious diseases transmission are one of the primary tools for
understanding the propagation of infectious diseases among plant, animal, or human populations \cite{hethcote2000mathematics,keeling2008modeling,anderson1992infectious}. Understanding how spreading dynamics are affected by individual-level transmission characteristics and large-scale properties of interactions aids endeavors to control and mitigate epidemics, making it critical for the public health and security.

In addition to their critical role in public health decision making \cite{knight2016bridging}, infectious disease models are appealing from complex systems perspective. Take for instance the Susceptible-Infected-Susceptible (SIS) model \cite{anderson1992infectious}, where each
individual in the population is either `\emph{Susceptible}' or
`\emph{Infected}'. The SIS model simply states that susceptible
individuals may become infected when interacting with infected individuals, and
infected individuals will become susceptible immediately after recovery. Rich dynamics of the SIS model, such as the phase transition observed between fast die-out of infections and long-term epidemic persistence \cite{pastor2015epidemic}, exemplify the ability of simple individual-level interactions to give rise to emergent phenomena.

Understanding disease transmission dynamics in human social networks is particularly challenging \cite{moran2016epidemic}, partly because humans take preventive measures and alter their interactions in response to disease spreading \cite{bish2010demographic}, which subsequently change the course of the spreading \cite{Book2013HumanBehavior}.  As such, coupled modeling of behavioral change and infection transmission dynamics has seen significant attention recently \cite{funk2010JRSI, Book2013HumanBehavior,verelst2016behavioural,wang2015coupled}. Medical treatments, quarantines, illness management practices, and individual preventive behaviors are a few examples of ways society works to reduce disease spreading.



Common preventive behaviors of individuals to the emergence of an epidemic are (1) adopting hygiene/pharmaceutical actions such as wearing a mask, using condoms, improving bodily/environmental cleanliness, and receiving vaccinations, and (2) altering contacts to avoid infection. In the first case, individuals are intending to reduce the probability of infection by cleansing themselves and their environment --- or at least placing barriers between the two \cite{chen2006JMB, funk2010JTB, funk2009NAS, Nicola2011PLOS, poletti2009JTB2, FaryadCDC11SAIS}. In the second case, when individuals change who they come in contact with, the fundamental topology of the network itself is changing. As individuals remove certain contacts with people, while possibly creating new ones, the structural paths available to dynamic processes are being altered, resulting in rich dynamic interplay between network topology and the spreading process on top of it \cite{reluga2010game, Mina2011JTB,gross2008JRSI, gross2006PRL, marceau2010PRE,
risau2009JTB, demirel2012X, van_segbroeck2010PLoS, Ves2011SciRep}.

Existing approaches to incorporate preventive behaviors in mathematical models of infectious diseases fall into two general categories. First approach incorporates the effect of preventive behaviors directly into disease model parameters \cite{fenichel2011adaptive,liu2015endemic,brauer2011simple,li2015bifurcation,morin2014disease,xiao2012sliding,paarporn2015epidemic}. The second approach introduces additional dynamic states into a disease model to explicitly distinguish those who have adopted a preventive behavior from those who have not \cite{funk2010JTB,FaryadCDC11SAIS,sahneh2012SR,misra2011effect,samanta2013effect,misra2015stability,wang2015interaction}. One example of individual-based models taking the second approach is the susceptible--alert--infected--susceptible (SAIS) framework, first introduced in \cite{FaryadCDC11SAIS}.

The SAIS framework adds an `\emph{Alert}' state to the networked SIS model of \cite{van2009TN}. The alert state represents individuals who (similar to susceptible individuals) can potentially become infected, but has adopted a preventive behavior. In the original SAIS model \cite{FaryadCDC11SAIS}, alert individuals have a lower infection rate compared to the susceptible individuals, and susceptible individuals could become alert in presence of infection among their local contacts. The lower infection rate of alert individuals would correspond to their type-1 preventive behaviors (such as wearing masks or using condoms). This model predicts possibility of total eradication of an epidemics through preventive behaviors\cite{sahneh2012SR}. In a subsequent study \cite{SahnehCDC12}, authors considered an information-dissemination network as an alternative alerting mechanism, and proposed the optimal design solution for an information-dissemination network based on eigenvector centralities \cite{bonacich2007some} in the contact network graph. The SAIS model has been further explored in \cite{Preciado2013SAIS,shakeri2015optimal,juher2015analysis}.

In this paper, we introduce the AC-SAIS model, where AC stands for `\textbf{A}daptive \textbf{C}ontact', to model a scenario in which individuals change their contact neighborhood upon becoming alert. More specifically, each susceptible individual $i$ is in contact with a given set of individuals $(\mathcal{N}^S_i)$, and when she becomes alert, she switches to another set of individuals $(\mathcal{N}^A_i)$. We will use the terms \emph{default neighborhood} and \emph{adapted neighborhood} to distinguish the two. In our model, we assume both of these neighborhoods are known {\it a priori}. Yet, we do no make any restrictive assumptions on these neighborhood sets and deliver our results in the most generic setup. In practice, the default and adapted neighborhood sets might be closely related. For example, in a social distancing scenario \cite{maharaj2012controlling}, the adapted neighborhood would be a subset of the default neighborhood. Social distancing is not the only possible scenario of contact adaptation. In the context of sexually transmitted infections, for example, when a person is notified that one of his sexual partners is infected, in response, he may abandon all or some of his set of partners and seek partnership from a new venue.



When nodes adapt their contacts to a neighborhood constituting a more robust network, one might intuitively expect that the robustness of the network against epidemic spreading increases monotonically with the contact adaptation rate. This is true in the case of social distancing (where the alert neighborhood is a subset of the default ones) as it always help mitigating epidemic spreading, and the faster the social distancing is implemented, the better. However, when the set of adapted contacts of an individual is not restricted to be a subset of their default contacts, the network robustness against epidemic spreading can be a non-monotone function of the contact adaptation rate. Indeed, our model detects a counter-intuitive scenario where adaptively changing contacts may adversely lead to lower network robustness against epidemic spreading if the adaptation is not fast enough. 

From dynamical systems perspective, this study contains several contributions. First, we propose a novel state-dependent switching network framework and show that a multilayer-network \cite{Kivela2014} formulation can be successfully employed. Second, we develop an original result of nonlinear Perron-Frobenius theory, where we find necessary and sufficient conditions for existence and uniqueness of a strictly positive eigenvector for the class of non-negative, concave maps. We apply this tool to find the epidemic threshold for our AC-SAIS model. Furthermore, we introduce a novel notion of connectivity for multilayer networks, which is novel for the new research field of multilayer networks.

The rest of the paper is organized as follows: After the literature review in Section \ref{Sec: LitRev}, Section \ref{Sec: Background} introduces a novel notion of multilayer connectivity and an original result for nonlinear Perron-Frobenius theory, which are pivotal for the subsequent modeling and analysis. Section \ref{Sec: Model} develops the AC-SAIS model, showing that the proposed adaptive contact can be equivalently modeled by multilayer networks. Analyses in Section \ref{Sec: ETE} are followed by numerical experiments in Section \ref{Sec: Numerical Sim}. Several proofs to theorems and lemmas are omitted for the sake of brevity, and can be found in the Supplemental Materials of this article.

\section{Literature Review\label{Sec: LitRev}}

Typical approaches to modeling spreading processes on networks consider network
topologies as independent of individual node states, such is the case when
nodes retain the same set of contacts regardless of whether or not they, or
their neighbors, are infected. This assumption is made for simplicity's sake
and is not representative of real world networks; especially in regards to
social networks where a person's contacts are in constant fluctuation. The notion
of state-dependent topologies is especially poignant in the context of disease
dynamics where a person will adjust who they come in contact with when in the
presence of an infection.
The extent to which this occurs can vary greatly --- from removing a single
contact to completely changing all of them --- depending on the perceived
severity of an infection.

Several formulations of adaptive contact exist in the literature of infectious disease modeling, including: 1) social distancing \cite{valdez2012intermittent}, where healthy individuals lower their contact with the rest of the population, 2) delete-and-reactivate \cite{VanMieghem2013PRE}, where healthy break their contact with infected population and reactivate after some time, 3) rewiring \cite{juher2013outbreak,dong2015can}, where healthy break their contact with infected population and create new links with healthy members \cite{gross2006PRL} or any other randomly chosen individual \cite{risau2009JTB}.


Altering the local contacts can have a strong effect on disease dynamics, which in turn influences the contact adaptation process; a complicated mutual interaction between a time varying network topology and the dynamics of the nodes emerges. For example, Gross et al. \cite{gross2006PRL} presented a model where susceptible individuals rewire their links from other infected individuals toward susceptible ones in an SIS model, resulting in the formation of two loosely connected clusters. Several researches have built on this model: Marceau et al. \cite{marceau2010PRE} additionally include the infection state of its neighbors in the node information. Risau et al. \cite{risau2009JTB} rewire susceptible individuals from infected neighbors to random nodes, which in some cases completely suppresses epidemic spreading.


Most of contact adaptation schemes have been implemented for well-mixed populations or random network models of physical interactions. Studies that work with generic graphs as their contact network are scarce in the literature. Among a few existing research endeavors is the Adaptive-SIS (ASIS) model developed by Guo et al. \cite{VanMieghem2013PRE}, who studied an SIS epidemic model where contacts between susceptible and infected nodes are removed at some rate and reactivated later. They showed the epidemic threshold increases as a function of the link removal rate, while the network topology exhibits binomial-like degree distribution, assortative mixing, and modularity. This approach was rigorously extended by Ogura and Preciado \cite{Ogura2016}, who additionally considered heterogeneous node and edge parameters, as well as a method for optimizing adaptation rates to mitigate epidemic outbreaks. This approach of adaptation for generic graphs considers a dynamic equation for the edge weights which is coupled with the epidemic model. Another approach would be through the notion of switching networks in dynamical systems.

A switching contact network is defined as a set of distinct networks where the ``active" network at any given time is determined by some switching signal. More precisely, we denote a switching network $G(t)=(V,E^{s(t)})$, where
$s(t):\mathbb{R}\rightarrow\{1,2,...,q\}$ is a signal that determines which of the $q$
networks are active at time $t$. Usually this signal is external and
independent of the system states. For example, a common approach is to
consider $s(t)$ as a Markov process independent from the disease states \cite{ogura2015disease,ogura2016stability}. The collection of possible edge sets $\mathcal{E}=\{E^1,E^2,...,E^q\}$ may be given {\it a priori} as in \cite{ogura2015disease}, or they might be generated from local processes as in \cite{ogura2016stability}.  In the latter, Ogura and Preciado considered a base graph with $|E|$ edges where each edge can become active or inactive according to an externally defined Markov process, leading to an overall $2^{|E|}$ possible configurations for the switching contact network. We can also think of a more
complex situation where the switching signal is dependent on the system
states. In this way, the topology of the active network determines the evolution of the dynamic process and in turn, the state of the process itself signals network switching. Here lies our proposed contact adaptation scheme.

We consider a class of switching networks where the neighborhood set of each
node depends on
the state it occupies. Specifically, each node $i$ has one of two contact sets $\mathcal{N}^S_i$ and $\mathcal{N}^A_i$, depending on whether is is `susceptible' or `alert'. Therefore, for a network of size $N$, the entirety of the switching network is composed of $2^{N}$ separate topologies. In this case, not only the network state-space size exponentially increases by $N$, but also the switching signal depends on the collective system state. 






\section{Fundamental Concepts and Tools\label{Sec: Background}}
Before diving into the modeling and analysis, we first start with a novel notion of connectivity for multilayer networks and an original results for a class of nonlinear Perron--Frobenius problems that will facilitate the subsequent developments in this paper.

\subsection{Nonlinear Perron Frobenius}\label{Sec: NLPF}
The classical Perron-Frobenius theorem \cite{van2010graph} concerns the eigenvalue problem $Ax=\lambda x$ for a nonnegative and irreducible matrix $A$. Let $\mathbb{R}_{+}^{n}$  be the non-negative cone in the $n-$dimensional Euclidean space,
\[
\mathbb{R}_{+}^{n}=\lbrace x\in \mathbb{R}^{n}\vert x_{i}\geq 0 ~ ~\text{for}~~ 1\leq i\leq n\rbrace.
\]
Assuming $x,y \in \mathbb{R}^{n}_{+}$, here $x\preceq y$ ($x\prec y$) means $x_{i}\leq y_{i}$ ($x_{i}<y_{i}$) for $1\leq i\leq n$ and $x\precnsim y$ denotes $x\preceq y$ but $x\neq y$.
A matrix $A=[a_{ij}]$ is called non-negative if all of its entries are either positive or zero. We can construct a graph $G(A)$ associated with $A$ such that the edge $(i,j)$ exists if $a_{ij}>0$. The matrix $A$ is irreducible if and only if its associated graph $G(A)$ is strongly connected. The classical Perron-Frobenius theorem may be stated as the following:

\begin{theorem}[Perron--Frobenius Theorem \cite{van2010graph}]\label{Th: PF}
Let $A$ be a nonnegative, irreducible matrix. Then $A$ has a positive eigenvalue $\lambda_1>0$ which has multiplicity one and any  eigenvalue of $A$ has a magnitude smaller than or equal to $\lambda_1$. Furthermore the eigenvector $\boldsymbol{v}_1$ corresponding to $\lambda_1$ is strictly positive (i.e., $\boldsymbol{v}_1\succ 0$) and is the only eigenvector of $A$ in the nonnegative cone.
\end{theorem}

From mappings perspective, the classical Perron--Frobenius theory concerns solutions to the eigenvalue problem $F(x)=\lambda x$ where $F(x)=Ax$ is a linear self-map of the non-negative cone. By ``self-map of the non-negative cone,'' we mean that $F:\mathbb{R}_{+}^{n}\rightarrow\mathbb{R}_{+}^{n}$ maps the non-negative cone to itself. But what if the map $F(x)$ is not linear? Can we still get powerful results for nonlinear maps analogous to the Perron--Frobenius theorem? The whole area of the nonlinear Perron--Frobenius theory \cite{krasnoselskij1964positive,per1,per2,krause2001concave,nussbaum1999generalizations} seeks answer to these questions. A thorough review of nonlinear Perron--Frobenius theory is out of the scope of this paper. In short, results are usually more limited in that existence, uniqueness, or strictly positivity of an eigenvector is seldom guaranteed unless under restrictive assumptions on the nonlinear map. 

The following properties are among the possibilities to relax the linearity assumption for the non-negative map $F$. Note that the linear map $F(x)=Ax$ with non-negative matrix $A$ has all of these properties. 



\begin{definition}
Assume $F:\mathbb{R}_{+}^{n}\rightarrow\mathbb{R}_{+}^{n}$ is a self-map of nonnegative cone.  We say $F$ is 
\begin{enumerate}
	\item {\em homogeneous}, if for any $x\succeq0$ and $c\geq0$, $F(cx)=cF(x)$,
	\item {\em concave}, if $F(\theta x+(1-\theta)y)\succeq \theta F(x)+(1-\theta)F(y)$
	for all $x,y\succeq0$ and $0\leq \theta \leq 1$,
    \item {\em super-additive}, if $F(x+y)\succeq F(x)+F(y)$
	for all $x,y\succeq0$,
    \item {\em monotone}\footnote{Sometimes, this property is referred to as {\em order-preserving}.}, if $F(y)\succeq F(x)$ for all $y\succeq x \succeq0$.
\end{enumerate}
\end{definition}

The homogeneity property indicates that if $x^*\succeq 0$ is an eigenvector of $F$, so is $cx^*$ for any $c\geq 0$. Furthermore, the following lemma indicates that the class of homogeneous, concave self-maps of the non-negative cone is a special case of homogeneous, monotone maps.

\begin{lemma}\label{Lemma: Monotonicity}
	If $F:\mathbb{R}_{+}^{n}\rightarrow\mathbb{R}_{+}^{n}$ is a homogeneous, concave map of the non-negative cone, then $F$ is also monotone and super-additive. 
\end{lemma}

Several results in the literature concern the more general class of homogeneous, monotone maps \cite{per1,per2}.  While existence and strict positivity of an eigenvector can be proved for this class of maps, uniqueness cannot be guaranteed without quite restrictive assumptions \cite{per1}. For example\footnote{This example is from \cite{per1}.}, for the homogeneous, monotone function $F(x)=[\max\{x_1,\frac{x_2}{2}\},\max\{\frac{x_1}{2},x_2\}]^T$, any vector $[x_1,x_2]^T\in \mathbb{R}_+^2$ with $\frac{x_1}{2}\leq x_2 \leq 2x_1$ is an eigenvector of $F$ with eigenvalue $\lambda=1$. On the contrary, existence and strict positivity of a unique eigenvector can be proved for the special class of homogeneous, concave maps.

The nonlinear map of interest in this paper falls in the special class of homogeneous and concave maps. Therefore, we focus on this class of nonlinear maps and develop a new result.

So far, we relaxed the linearity restriction by assuming that our nonlinear map is homogeneous and concave. The next question is what would be the counter part to irreducibility of $A$ in the linear map $F(x)=Ax$ for a homogeneous, concave map. For homogeneous, concave maps, Krause \cite{krause2001concave} proposes the following condition:

 \begin{definition}[Krause  \cite{krause2001concave}, \S 3] \label{def: KrauseIrred}
 We say the homogeneous, concave self-map $F:\mathbb{R}_{+}^{n}\rightarrow\mathbb{R}_{+}^{n}$ satisfies condition\footnote{In Krause \cite{krause2001concave}, authors refer to this condition as being {\em irreducible}. We choose to avoid this term to avoid any confusion with other notions that tend to extend irreducibility of linear maps to nonlinear domain.} {\bf C1}  in $\mathbb{R}_{+}^{n}$ if for any non-empty subset $\emptyset\neq J\subsetneq \{1,...,n\}$, there exists $j\in J$ and $i\notin J$ such that $F_i(e_j)>0$, where $e_j$ is the $j-$th unit vector in $\mathbb{R}^n$ and $F_i$ denotes the $i-$th component of $F$.
 \end{definition}
 
Furthermore, Krause proves that condition {\bf C1} is a {\em sufficient condition} for existence and uniqueness of a positive eigenvector:
 
 \begin{theorem}[Krause  \cite{krause2001concave}, Theorem 13] \label{Th: Krause}
For the self-map $F:\mathbb{R}_{+}^{n}\rightarrow\mathbb{R}_{+}^{n}$, which is concave, homogeneous, and satisfies condition {\bf C1}, the equation $F(x)=\lambda x$ has a strictly positive solution $x=x^*\succ 0$, $\lambda=\lambda^*>0$, and $x^*$ is the only eigenvector in the non-negative cone (up to scaling).
 \end{theorem}

We argue that the condition {\bf C1} for the notion of irreducibility in \cite{krause2001concave} may be restrictive, and same strong results would be still valid under a more relaxed condition. Indeed, the nonlinear map of our interest in this paper may not satisfy the condition {\bf C1} in Definition \ref{def: KrauseIrred}.

To illustrate, suppose $n=3$ and the nonlinear map is $F(x)=[\min\{x_2,x_3\},x_1+x_3,x_1+x_2]^T$. This map is both homogeneous and concave. However, it does not satisfy condition {\bf C1} of \cite{krause2001concave} stated in Definition \ref{def: KrauseIrred}. To test this, let $J=\{2,3\}$; no $j\in J$ leads to $F_1(e_j)>0$ because $F(e_2)=e_3$ and $F(e_3)=e_2$. However, this map has a unique, strictly positive eigenvector $x^*=[\frac{1}{1+2\lambda^*},\frac{\lambda^*}{1+2\lambda^*},\frac{\lambda^*}{1+2\lambda^*}]^T$ and $\lambda^*=\frac{1+\sqrt{5}}{2}$ with $||x^*||_1=1$. Another example is $F(x)=[\frac{x_2x_3}{x_2+x_3},x_1+x_3,x_1+x_2]$.  Again, $F(e_2)=e_3$ and $F(e_3)=e_2$, so it does not satisfy condition {\bf C1}.  However, this map has a unique, strictly positive eigenvector $x^*=[\frac{1}{1+4\lambda^*},\frac{2\lambda^*}{1+4\lambda^*},\frac{2\lambda^*}{1+4\lambda^*}]^T$ and $\lambda^*=\frac{1+\sqrt{3}}{2}$ with $||x^*||_1=1$.

\begin{definition} \label{C2}
We say the homogeneous, concave self-map $F:\mathbb{R}_{+}^{n}\rightarrow\mathbb{R}_{+}^{n}$ of the non-negative cone satisfies condition {\bf C2} in $\mathbb{R}_{+}^{n}$ if for any choice of $\emptyset\neq J\subsetneq \{1,...,n\}$, there exists $i\notin J$ such that $F_i(e_J)>0$, where $e_J$ is defined as $e_J\triangleq\sum_{j\in J}e_j$.
\end{definition}

The example function $F(x)=[\min\{x_2,x_3\},x_1+x_3,x_1+x_2]^T,$ which does not satisfy condition {\bf C1}, does satisfy {\bf C2}. For instance, selecting $J=\{2,3\}$ yields $F_1(e_J)>0$ because $F(e_J=[0,1,1]^T)=[1,1,1]^T$. The following lemma proves that {\bf C2} is indeed less restrictive than {\bf C1}.

\begin{lemma}\label{Lemma: C1C2}
A homogeneous, concave self-map $F$ of the nonnegative cone that satisfies condition {\bf C1} also satisfies condition {\bf C2}.
\end{lemma}

We would like to emphasize that there is nothing special about usage of $e_J$ in Definition \ref{C2}. The following lemma shows that any vector that has positive values on elements corresponding to $J$ and is zero on other elements would be equivalently applicable. 

\begin{lemma}\label{Lemma: allx}
For any choice $x\succ 0$, we have $F_i(x\circ e_J)>0$ if and only if $F_i(e_J)>0;$ where the symbol $\circ$ denotes the Hadamard (entry-wise) multiplication.
\end{lemma}

In the linear domain, we know that if a non-negative matrix $A$ is irreducible, the matrix $A+cI$ is primitive for any $c>0$ \cite[Theorem 9]{wood2004always}, and vice versa. How would be the extension of this idea to the nonlinear domain? 
First, let us precisely define a primitive map.

\begin{definition}
The self-map $H:\mathbb{R}_{+}^{n}\rightarrow\mathbb{R}_{+}^{n}$ of the non-negative cone is called  {\em primitive} if there exists $M$ such that $H^m(x)\succ 0$ for all $m\geq M$ and $x\succnsim 0$. Here, $H^{m}$ denotes the $m-$th iterate of $H$, i.e., $H^m(x)=H(H^{m-1}(x))$ and $H^0(x)\triangleq x$.
\end{definition}

The following theorem states that $F$ satisfying {\bf C2} and $F(x)+cx$ being primitive are equivalent.

\begin{theorem}\label{Th: Primitive}
The map $F_c(x)\triangleq cx+F(x)$ with $c>0$ is primitive if and only if the homogeneous, concave self-map $F:\mathbb{R}_{+}^{n}\rightarrow\mathbb{R}_{+}^{n}$ of the non-negative cone satisfies condition {\bf C2}. 
\end{theorem}

The duality between $F$ satisfying {\bf C2} and $F(x)+cx$ being primitive leads to the main theorem in this paper:

\begin{theorem}\label{peron}
	Statements of Theorem \ref{Th: Krause} still holds if condition {\bf C1} is replaced with condition {\bf C2}. Furthermore, if $x^*\succ 0$ is a unique eigenvector of the homogeneous concave map $F$ in $\mathbb{R}_+^N,$ then $F$ must satisfy condition {\bf C2}. Moreover, iterations of $F_c(x)$ with $c>0$ converge to $x^*,$ i.e.,
	\begin{equation}
	\lim\limits_{k\rightarrow\infty} \bar{F}_c^k(x)=x^*,\text{ for all $x\succnsim 0$, and } \bar{F}_c(x)\triangleq\frac{F_c(x)}{||F_c(x)||}.\label{NPF_iteration}
	\end{equation}
\end{theorem}

Compared with Theorem \ref{Th: PF}, it is evident that results for the nonlinear Perron--Frobenius problem in case of homogeneous, concave maps are very strong; existence and uniqueness of a strictly positive eigenvector can be guaranteed. Our contribution to the theory of nonlinear Perron--Frobenius theory for homogeneous, concave maps is that we relaxed the \underline{sufficient} condition of \cite{krause2001concave} (through replacing {\bf C1} by {\bf C2}) and proved that this new\footnote{We got the inspiration for our definition of condition {\bf C2} for homogeneous, concave maps from a notion in \cite{per1} for homogeneous, monotone maps. Gaubert and Gunawardena \cite{per1} refer to the homogeneous, monotone map $F$ as {\em indecomposable} if for any choice of $\emptyset\neq J\subsetneq \{1,...,n\}$, there exists $i\notin J$ such that $\lim_{a\rightarrow \infty}F_i(r_J(a))=\infty$, where $r_J(a)$ is defined as $(r_J(a))_j=a$ if $j\in J,$ and $(r_J(a))_j=1$ otherwise. While this may look very similar (or perhaps equivalent) to condition {\bf C2}, we would like to point out a subtle difference which can be very consequential. Consider the function $F(x)=[\frac{4x_1^{\frac{1}{2}}x_2^{\frac{3}{2}}}{x_1+x_2}, x_1+x_2]^T$. This function is homogeneous, concave, and monotone. It does not satisfy condition {\bf C2} because for $J=\{2\}$ we get $F_1(e_J)=F_1(e_2)=0$. However, it falls in the category of indecomposable maps of \cite{per1} because $\lim_{a\rightarrow \infty}F_1(r_2(a))=\lim_{a\rightarrow \infty}\frac{4a^{\frac{3}{2}}}{1+a}=\infty$ and $\lim_{a\rightarrow \infty}F_2(r_1(a))=\lim_{a\rightarrow \infty}1+a=\infty$. The nonlinear eigenvalue problem for this function gives two eigenvectors in $\mathcal{R}_+^2$, namely, $x_1^*=[1,1]^T$ with $\lambda_1=2$, and $x_2^*=[0,1]^T$ with $\lambda_2=1$; which is consistent with the fact that it does not satisfy condition {\bf C2}.} condition is also the \underline{necessary} condition for uniqueness of the eigenvector in the non-negative cone. 

\subsection{Multilayer Networks and an Algorithmic Notion of Connectivity}\label{Sec: Multilayer}

Graph theory is the mathematics of networks. In graph theory, a directed graph is formally defined as an ordered pair $G=(V,E)$, where $V$ is the set of nodes and $E\subset V\times V$ is the set of ordered pairs of nodes representing their directed relation. We say node $j$ is a neighbor of node $i$, if $(i,j)\in E$. The set $\mathcal{N}_i=\{j|(i,j)\in E\}$ denotes the neighbors of node $i$. A path $(v_0=i,v_1,...,v_{l-1},v_l=j)$ of length $l$ is an ordered tuple of edges than connects $i$ to $j$, i.e., $(v_{k-1},v_k)\in E$. A directed graph is strongly connected if there exists a path between all ordered pair of nodes in the network\cite{van2010graph}.

Several natural and technological systems show complex patterns of interactions among their heterogeneous entities. To capture the complexities of such systems, the network science community has recently shown substantial interest in the notion of multilayer networks \cite{Kivela2014,boccaletti2014structure} and developing proper mathematics for them beyond the classical graph theory \cite{de2013mathematical}.

In this paper, we denote a \textit{multilayer} network\footnote{In some literature, this may be referred to as a \emph{multiplex} network.} as an ordered tuple $\mathcal{G}=(V,E_A,E_B)$ where nodes in $V$ are connected through two link types $E_A$ and $E_B$. Corresponding to the multilayer network $\mathcal{G}$, we define $G_A=(V,E_A)$ and $G_B=(V,E_B)$ as the \textit{layers} of $\mathcal{G}$. Motivated by the notion of strong connectivity for directed graphs, we propose a novel notion of connectivity for multilayer networks in the following.

Our proposed notion of multilayer connectivity, which from now on we will refer to it as \textit{M--connectivity}, has an algorithmic definition. To motivate and acquaint our definition to the reader, we first point out a straight-forward property of simple strongly connected graphs. Suppose for the graph $G=(V,E)$ we have an arbitrary partition $\mathcal{P}$ of the node set $V$, i.e., members of $\mathcal{P}$ are non-empty disjoint subsets of $V$ that cover $V$, more precisely:
\begin{equation*}
\begin{split}
&(1)\, \emptyset\notin \mathcal{P},\\
&(2)\, I\cap J=\emptyset \text{ for any } I\neq J\in \mathcal{P}\\
&(3)\bigcup_{I\in \mathcal{P}}I=V.
\end{split}
\end{equation*}
We can build a graph $\boldsymbol{G}=(\mathcal{P},\mathcal{L})$, where the partition
set $\mathcal{P}$ is the node set of $\boldsymbol{G}$. Note that each node $I\in \mathcal{P}$ of $\boldsymbol{G}$ is a partitioning subset of $V$. As such, to avoid possible confusion, we will refer to nodes of $\boldsymbol{G}$ as \emph{hypernodes} from now on. We assign a directed link from one hypernode $I\in\mathcal{P}$ to another hypernode $J\in\mathcal{P}$, if there is a node $i\in I$ of $G$ that is connected to a node $j\in J$, i.e., $(i,j)\in E$. Trivially, yet importantly, strong connectivity of $G$ implies strong connectivity of $\boldsymbol{G}$. For a multilayer network $\mathcal{G}$, we use a related notion to define connectivity\footnote{We have been inspired by the notions of indecomposability for nonlinear maps and the method of aggregated graphs  from Gaubert and Gunawardena \cite[\S 1.3 \& \S 3.4]{per1}. }. The main difference is that connection among subsets must be through both layers. Following provides a formal definition.

For a multilayer network $\mathcal{G}=(V,E_A,E_B)$, we iteratively build graphs $\boldsymbol{G}^k=(\mathcal{P}_k,\mathcal{L}_k)$, starting with $\boldsymbol{G}^0=(\mathcal{P}_0,\emptyset)$, where $\mathcal{P}_0=\{\{1\},\{2\},...,\{N\}\}$ is the trivial partition of $V$ singletons. From the graph $\boldsymbol{G}^{k-1}$, we build $\boldsymbol{G}^k=(\mathcal{P}_k,\mathcal{L}_k)$ in the following way:

\textbf{Step 1:} Define the hypernode set $\mathcal{P}_k$ of cardinality equal to the number of strongly connected components of $\boldsymbol{G}^{k-1}$ where each element $I\in \mathcal{P}_k$ groups one and only one strongly connected component of $\boldsymbol{G}^{k-1}$ (i.e., $I$ is the union of all the hypernodes in that strongly connected component). Note that, doing so, the hypernode set $\mathcal{P}_k$ always denotes a partitioning of the node set $V$.

\textbf{Step2:} We assign the directed link $(I,J)\in \mathcal{L}_k$ if at least one \textbf{single} node in $I$ is connected to $J$ through both layers \textbf{simultaneously}\footnote{Note that this is different from, $\exists i_1,i_2\in I~~\text{s.t.}~~(i_1,j_{1})\in E_A,\,(i_2,j_{2})\in E_B~~\text{for some}~~j_{1},j_{2}\in J$, which basically indicates that both individual layers $G_A$ and $G_B$ are strongly connected.}, i.e.
\begin{equation*}
\begin{split}
\mathcal{L}_k=\bigg\{(I,J)\in&\mathcal{P}_k\times\mathcal{P}_k\Big\vert \exists i\in I~~\text{s.t.}~~(i,j_{1})\in E_A,\\
&~~~~~(i,j_{2})\in E_B~~\text{for some}~~j_{1},j_{2}\in J\bigg\}.
\end{split}
\end{equation*}

Figure \ref{scg} illustrates the iterative procedure explained above. 

\begin{definition}\label{def: sc}
	A multilayer network $\mathcal{G}=(V,E_A,E_B)$ is \emph{M-- connected}, if starting with with $\boldsymbol{G}^0=(\mathcal{P}_0,\emptyset)$---where $\mathcal{P}_0=\{\{1\},\{2\},...,\{N\}\}$ is the trivial partition of $V$ singletons---and inductively building $\boldsymbol{G}^1,\boldsymbol{G}^2,...$ following Step 1 and Step 2 described above, there exists an iteration step $k_*$ such that $\boldsymbol{G}^{k_*}$ is strongly connected. 
\end{definition}

Intuitively, M--connectivity of $\mathcal{G}$ implies that if we split the node set $V$ into any two subsets $V_a$ and $V_b$, there is always a node in $V_a$ (resp. $V_b$) that is connected to $V_b$ (resp. $V_a$) through both edge types. A necessary condition for M--connectivity of $\mathcal{G}$ is that both individual layers $G_A$ and $G_B$ are strongly connected. Moreover, a sufficient condition for M--connectivity of $\mathcal{G}$ is that the intersection graph $G_c\triangleq(V,E_A\cap E_B)$ is strongly connected (because, $\boldsymbol{G}^1$ which is similar to $G_c$, will be already strongly connected).

\begin{figure}[h!]
	\centering
	\begin{subfigure}{1.5in}
		\includegraphics[width=1\textwidth]{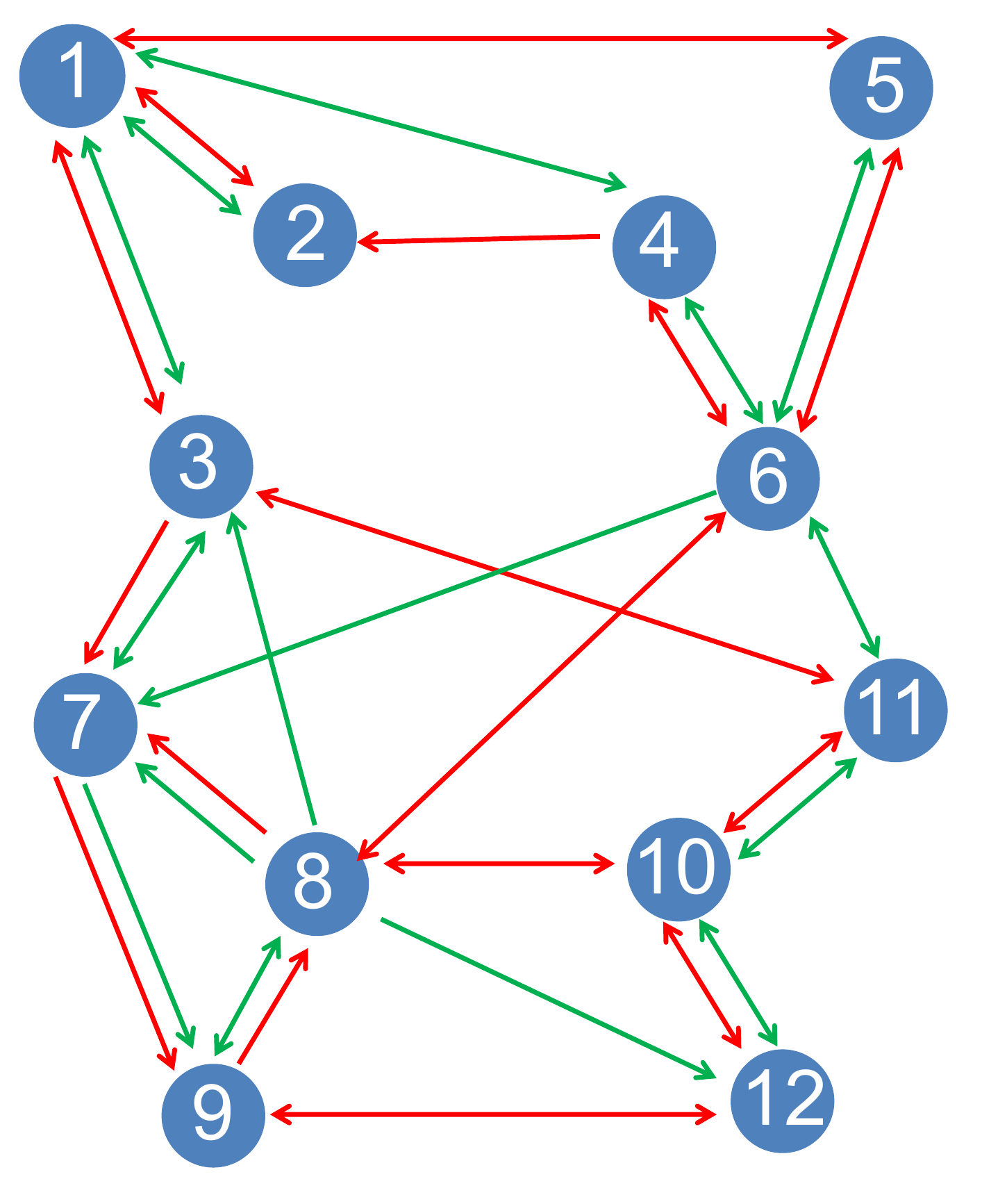}
		\caption{}
		\label{sca}%
	\end{subfigure} \hspace{0.23in}
	\begin{subfigure}{1.5in}
		\includegraphics[width=1\textwidth]{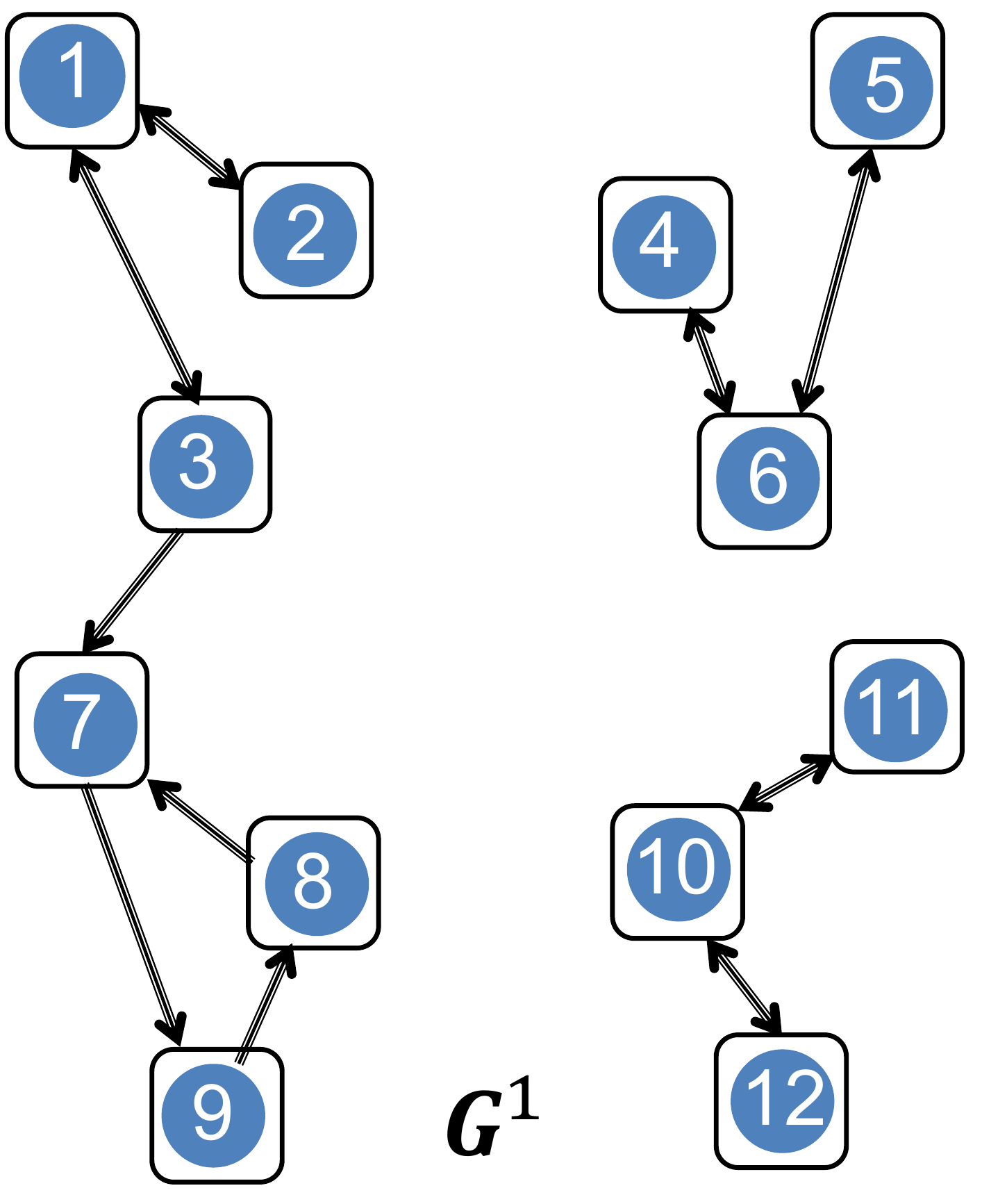}
		\caption{}
		\label{scb}%
	\end{subfigure}
	\begin{subfigure}{1.5in}
		\includegraphics[width=1\textwidth]{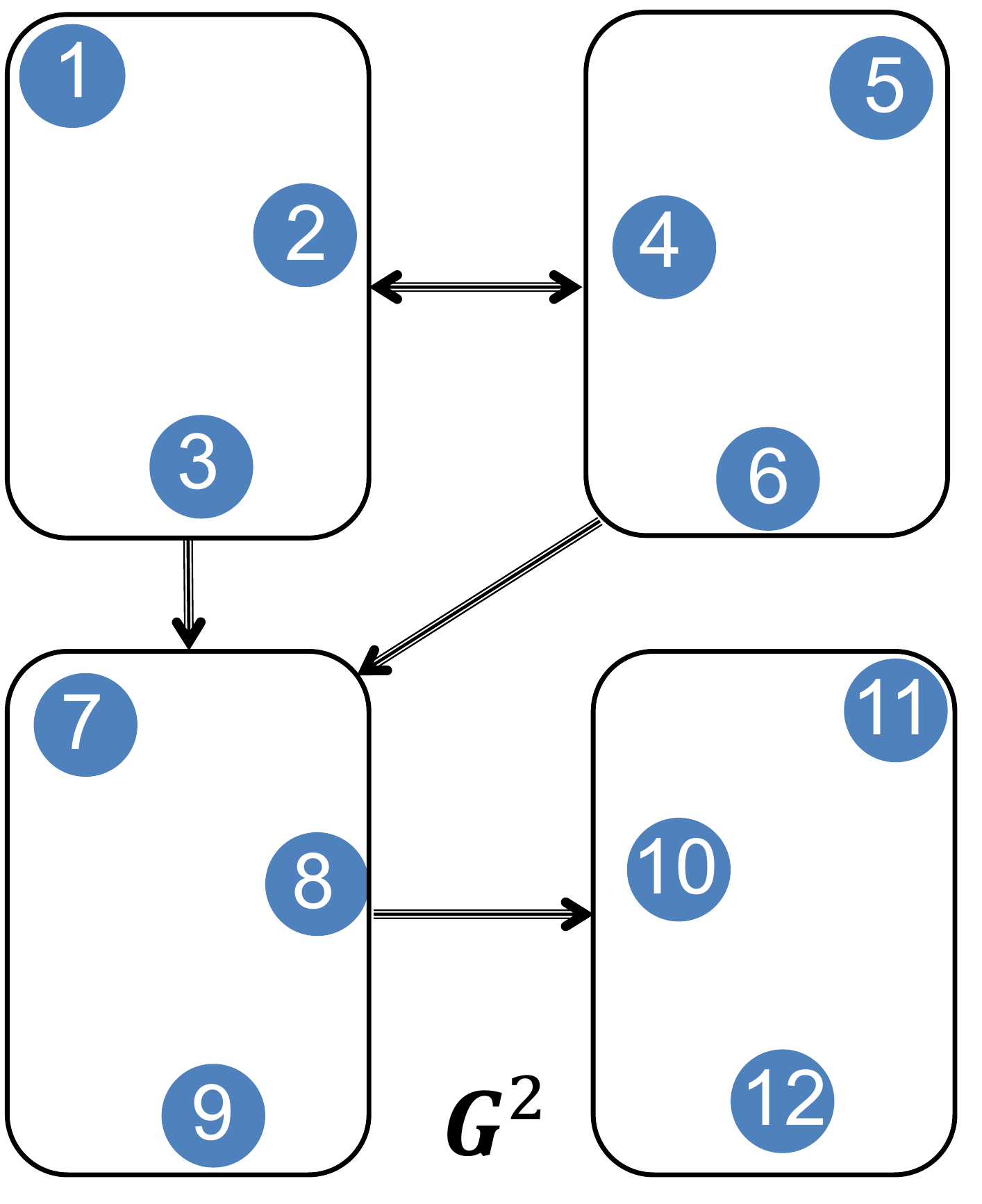}
		\caption{}
		\label{scc}%
	\end{subfigure} \hspace{0.05in}
	\begin{subfigure}{1.5in}
		\includegraphics[width=1\textwidth]{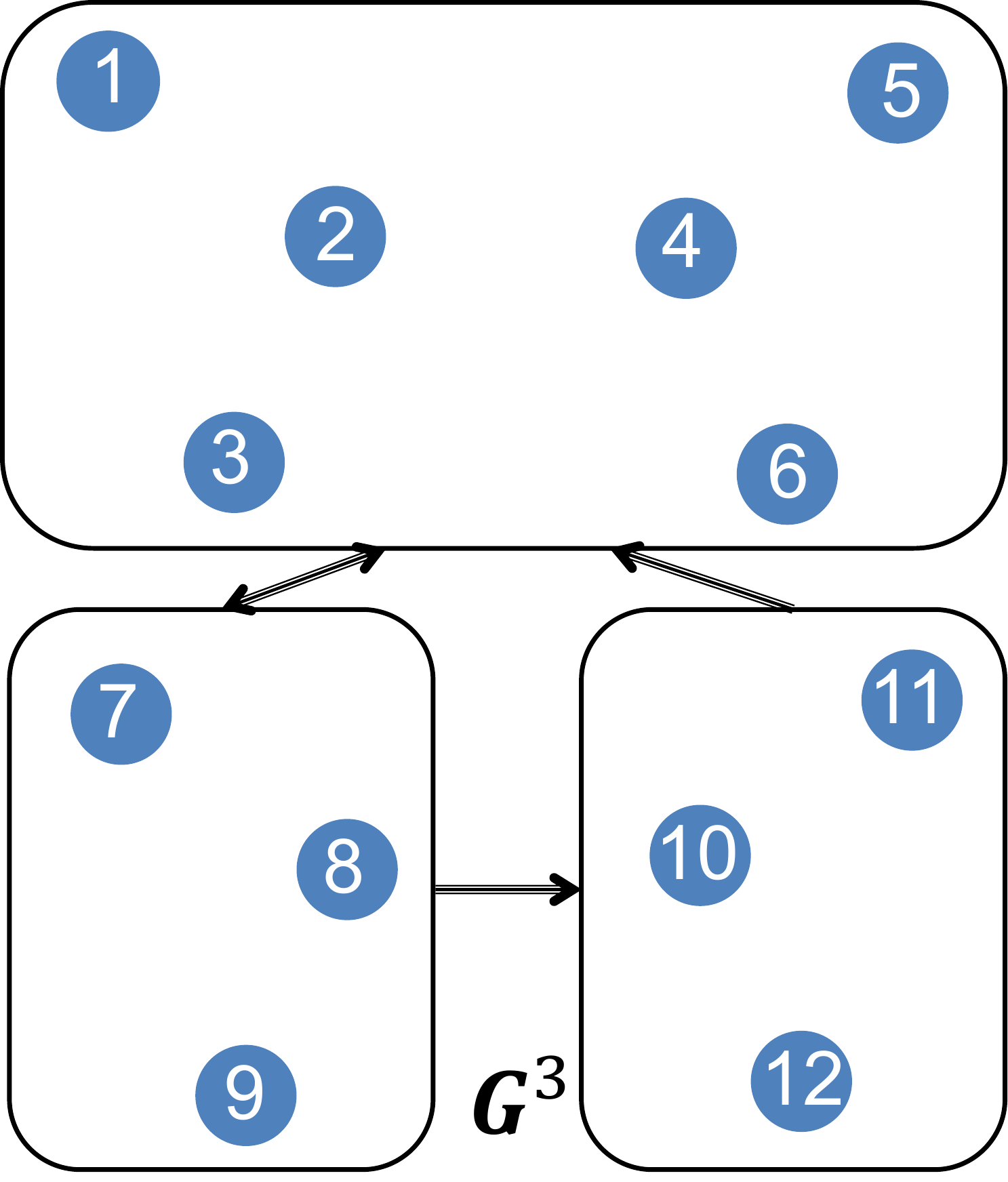}
		\caption{}
		\label{scd}%
	\end{subfigure}
	\caption{Example of an M-connected multilayer network according to Definition \ref{sc}. (a) In the two-layer graph, the red arrows represent $E_A$ edges and green arrows represent $E_B$ edges. (b) The first graph $\boldsymbol{G}^1$ has the hypernodes set $\mathcal{P}_1=\{\{1\},\{2\},...,\{12\}\}$, and its links are the intersection of $E_A$ and $E_B$ edges. Hypernodes are depicted by black squares, and links between them are shown by black arrows. The graph $\boldsymbol{G}^1$ is not strongly connected. (c) The second aggregate graph $\boldsymbol{G}^2$  has the strongly connected components of $\boldsymbol{G}^1$ as its hypernodes set $\mathcal{P}_2=\{\{1,2,3\},\{4,5,6\},\{7,8,9\},\{10,11,12\}\}$. The links among the hypernodes is according to Step 2 in Section \ref{Sec: Multilayer}. For example, the directed link $(\{4,5,6\},\{7,8,9\})\in \mathcal{L}_2$ is due to node $6$, a member of $\{4,5,6\}$, being connected to the hypernode $\{7,8,9\}$ through both layers (because, $(6,8)\in E_A$ and $(6,7)\in E_B.$). The graph $\boldsymbol{G}^2$ is not strongly connected either. (d) The third aggregated graph $\boldsymbol{G}^3$ groups strongly connected components of $\boldsymbol{G}^2$ as its hypernodes set $\mathcal{P}_3=\{\{1,2,3,4,5,6\},\{7,8,9\},\{10,11,12\}\}$. The graph $\boldsymbol{G}^3$ is strongly connected. Therefore, the two-layer network in (a) is M--connected according to Definition \ref{sc}.}
	\label{scg}%
\end{figure}

{\bf M--Connectivity and Condition C2.} Consider a multilayer network $\mathcal{G}=(V,E_A,E_B)$, where the node set are labeled from 1 to $N$, i.e., $V=\{1,2,...,N\}$. By defining a real valued vector $x:V\rightarrow\mathbb{R}_+^N$ on node set $V$, we show the relation between M--connectivity and condition {\bf C2} for functions $F(x)$.



\begin{theorem}\label{Th: MC2}
Associated with the multilayer network $\mathcal{G}=(V,E_A,E_B),$ where $V=\{1,2,...,N\}$, suppose a homogeneous, concave map $F$ of the non-negative cone is such that for any nontrivial subset $J$ of $V$ and $i\notin J$, $F_i(e_J)>0$ if and only if there exists $j_1,j_2\in J$ for which $(i,j_1)\in E_A$ and $(i,j_2)\in E_B$. Then, $F$ satisfies condition {\bf C2} if and only if the multilayer network $\mathcal{G}$ is M-connected.
\end{theorem}

\section{Model Development \label{Sec: Model}}
Before introducing our model, we first review a background on the networked SIS epidemic process.

\subsection{A Background on Networked SIS Model}\label{Sec: SIS}

Susceptible--infected--susceptible (SIS) model is a paradigmatic epidemic spreading model. In the SIS model, each individual is either \emph{susceptible} or \emph{infected}, and individuals are assumed to immediately become susceptible to the disease after recovery. SIS model is thus suitable for modeling sexually transmitted infections such as Gonorrhea and Syphilis \cite{keeling2008modeling}.

Classical compartmental epidemic models assume homogeneous (fully mixed) interactions among individuals. In networked epidemic models, interactions among individuals are explicitly modeled using a \emph{contact network}, represented by the graph $G=(V,E)$, where individuals are represented by nodes $V$ of a graph and possible interactions are the edges $E$ of a graph. Node $j$ is a neighbor of node $i$, denoted as $(i,j)\in E$, if she can infect node $i$ directly. We can also use weighted graphs to represent contact networks. Doing so, the weight value of a link would serve as a proxy for heterogeneity of contact levels among pairs of individuals. For example, if both nodes $j$ and $k$ are infected and $w_{ij}=2w_{ik}$, the likelihood that a susceptible node $i$ contracts the disease from node $j$ is double the likelihood of contracting it from node $k$. In this paper, we allow the contact graph be directed and weighted.

In the networked SIS model \cite{ganesh2005effect}, the state of node $i$ at time $t$ is denoted by $X_{i}(t)\in\{S,I\}$, where $X_{i}(t)=S$ if the node is susceptible or $X_{i}(t)=I$ if it is infected. In this model, a susceptible nodes becomes infected if it is exposed to an infected individual. Moreover, an infected individual recovers and becomes susceptible again after a recovery period. The infection and curing times are commonly assumed to have a memoryless property, leading to exponentially distributed time intervals in continuous time descriptions. More general time distributions are also possible and addressed in the literature to some extent \cite{lloyd2001viruses,neal2014endemic,van2013non,cator2013susceptible,ogura2016stability}.

The overall evolution of the nodes states are due to their interactions with each other. Hence, mathematical description of the SIS model requires utilization of the collective state $\boldsymbol{X}=[X_{1},...,X_{N}]$, which is the joint state of all $N$ nodes in the network. The network state is a continuous-time Markov process that undergoes transition over a space consisting of $2^{N}$ possible network states. In this description, we say an event has occurred if the state of a single node changes. Furthermore, the time interval for the event occurrence is exponentially distributed. This time interval can equivalently be described as the minimum of transition times of a set of statistically independent processes on node states, denoted by $X_i$, and pair states, denoted by $(X_i,X_j)$, as the following:
\begin{align*}
X_i&: I\rightarrow S \text{, for }i\in V,&T\sim exp(\delta),\label{SIS_node}\\
(X_{i},X_{j})&:(S,I)\rightarrow(I,I)\text{, if }(i,j)\in E,&T\sim exp(\beta w_{ij}),
\end{align*}
where $\delta$ and $\beta$ are called \textit{curing} and \textit{infection} rates, respectively, and $T$ represents the corresponding exponentially distributed transition duration.

Describing the network Markov process as competition among statistically independent nodal and edge-based transitions, similar to the above formulation of the SIS process, allows for a much more general framework for modeling networked epidemic processes (see, \cite{Sahneh2013TON}). We will use this approach to describe our adaptive contact epidemic model.

Finally, the Kolmogorov equation, which governs probability distribution of the SIS Markov process, is a system of $2^{N}$ coupled differential equations which is neither computationally nor analytically tractable for large number of nodes. Moment closure approximations \cite{gleeson2012accuracy,taylor2012markovian,van2009TN,Sahneh2013TON} or Monte Carlo simulations are thus necessary to study the networked SIS process. The SIS process shows a phase transition behavior where initial infections die out quickly for small values of $\beta/\delta$, while infections can persist in the network for long time (coined as \emph{metastable state}) for large values of $\beta/\delta$ \cite{Pastor-Satorras2015}. The critical value separating these regions is called the \textit{epidemic threshold}. As such, epidemic threshold suggests a measure of networks robustness against epidemic spreading. In this paper, whenever we say network $a$ is more robust against epidemic spreading than network $b$, we mean network $a$ has a larger value of epidemic threshold that network $b$.

\subsection{AC-SAIS Markov Model\label{Sec: MarkovModel}}

Consider a population of $N$ individuals, where each individual is either
\emph{susceptible}, \emph{alert}, or \emph{infected}. For each
individual $i\in\{1,...,N\}$, let the random variable $X_{i}(t)=S$ if the
individual $i$ is susceptible at time $t$, $X_{i}(t)=A$ if alert, and
$X_{i}(t)=I$ if infected. In the AC-SAIS model of this paper, contacts of a node depends on her state. Specifically, we
define $\mathcal{N}_{i}^{S}$ as the neighbors of node i when she is
susceptible, and $\mathcal{N}_{i}^{A}$ as her neighbors when she is alert. Associated with these neighborhood sets, we consider weight values $w^S_{ij}>0$ if $j\in\mathcal{N}_{i}^{S}$ and $w^A_{ij}>0$
if $j\in\mathcal{N}_{i}^{A}$ as a proxy for heterogeneity of the contact levels.

Four competing stochastic transitions describe the AC-SAIS model, as Fig. \ref{schematicsais} depicts:

\begin{enumerate}
	[leftmargin=*]
	
	\item \textbf{Infection of susceptible nodes:} A susceptible individual becomes infected from her infected neighbor (among $\mathcal{N}_i^S$) after an exponentially distributed random time duration with the infection rate $\beta$.
	\[
	\!\!\!(X_{i},X_{j}):(S,I)\rightarrow(I,I)\text{, if }(i,j)\in E_S,\quad T\sim exp(\beta w^S_{ij}).
	\]
	
	\item \textbf{Alerting of susceptible nodes:} A susceptible individual becomes alert from her infected neighbor (among $\mathcal{N}_i^S$) after an exponentially distributed random time duration with the {\em alerting} rate $\kappa$.
	\[
	\!\!\!(X_{i},X_{j}):(S,I)\rightarrow(A,I)\text{, if }(i,j)\in E_S,\quad T\sim exp(\kappa w^S_{ij}).
	\]
	
	\item \textbf{Infection of alert nodes:} An alert individual becomes infected due to having an infected neighbor among her switched neighborhood set $\mathcal{N}_i^A$ after an exponentially distributed random time duration with the infection rate $\beta$.
	\[
	\!\!\!(X_{i},X_{j}):(A,I)\rightarrow(I,I)\text{, if }(i,j)\in E_A,\quad T\sim exp(\beta w^A_{ij}).
	\]
	
	\item \textbf{Recovering of infected nodes:} An infected individual recovers to the susceptible state after an exponentially distributed random time duration with recovery rate $\delta$:
	\[
	X_i: I\rightarrow S \text{, for }i\in V,\qquad T\sim exp(\delta).
	\]
\end{enumerate}
	

\begin{figure}[ptb]
\includegraphics[scale=.37,trim={0 4.5cm 0 10mm}]{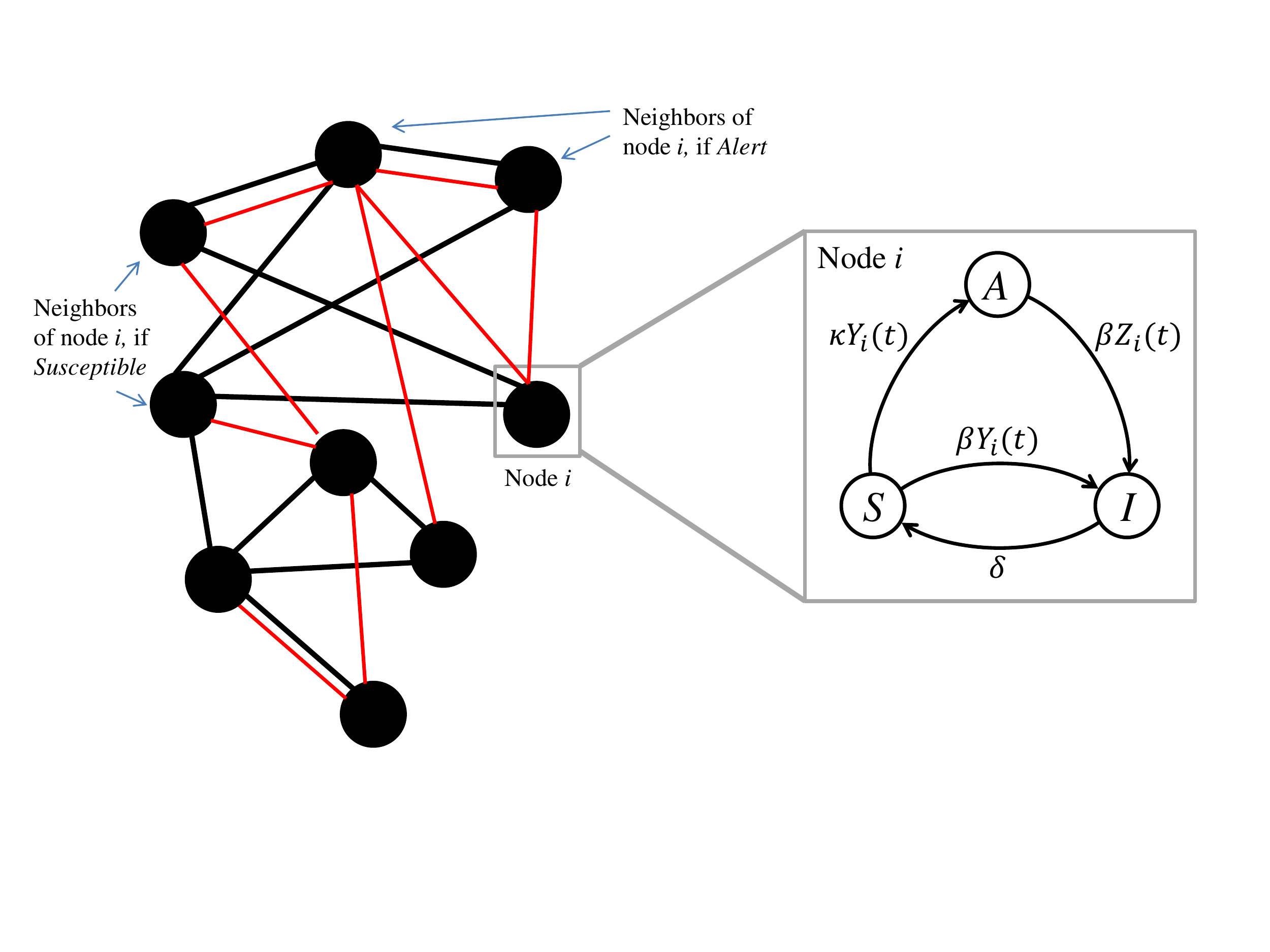} \caption{Schematic of the AC-SAIS
model. Black edges correspond to neighborhood $\mathcal{N}_{i}^{S}$ of susceptible 
node $i$, while red edges represent the neighborhood
$\mathcal{N}_{i}^{A}$ when node $i$ is in the alert state. Here, $\beta$, $\delta$, and
$\kappa$ are the infection rate, curing rate, and alerting rate, respectively.
$Y_{i}(t)$ is the number of infected neighbors of $i$ in $\mathcal{N}_{i}^{S}$
at time $t$ and $Z_{i}(t)$ is the number of infected neighbors of $i$ in
$\mathcal{N}_{i}^{A}$ at time $t$.}%
\label{schematicsais}%
\end{figure}

\vspace{.1 in}
{\bf A few remarks on the AC-SAIS model:} The disease dynamics component of the AC-SAIS model is according to the networked SIS model, elaborated in Section \ref{Sec: SIS}. Therefore, a representative example would be the spread of Syphilis or Gonorrhea for which the sexual contact network is well-defined and disease dynamics are SIS-type. In this scenario, the alerting process can be the result of a partner notification effort.

In the AC-SAIS model, we assume that if an alert individual never gets infected, she will remain in the alert state indefinitely. In other words, we do not consider an awareness decay process where alert individuals can transition to the susceptible state directly. In practice, we are assuming that the awareness decay process is so much slower than the disease dynamics that it becomes irrelevant for the disease spreading. Interested readers can refer to \cite{juher2015analysis} for analysis of an SAIS model with awareness decay.

The current setup of the AC-SAIS model only considers a type-2 preventive behavior of altering contacts, whereas the original SAIS model considered a type-1 preventive behavior by assuming a lower infection rate for alert individuals. It would be possible to also incorporate type-1 behaviors in the AC-SAIS model by lowering the infection rate for the alert individuals to $\beta_a<\beta$. In order to isolate the role of network adaptation, we do not change the infection rate in this study.

The contact alteration scheme in the AC-SAIS model assumes the contacts set of an individual only depends on her own state; it is the default set when susceptible, and the adapted set when alert. Particularly, the contact set of an alert individual is fixed and is independent of the health state of those contacts. Such contact adaptation scheme is most sensible when the identity of infected contacts are not known to the individual. For example, in the context of sexually transmitted diseases and partner notification, the identity of the infectious patient is not revealed to their partners. So, a subsequent contact adaptation may not necessarily lead to definite avoidance of infectious partners.

\subsection{An Equivalent Multilayer Representation\label{Sec: MultLayer}}


From a networked dynamical system perspective, the network topology in the AC-SAIS model is time-varying and switches among $2^{N}$ different possibilities because each node $i$ may adopt one of the two neighborhood sets $\mathcal{N}_i^S$ and $\mathcal{N}_i^A$. However, the AC-SAIS model can be equivalently interpreted as a spreading process on a two-layer network. The AC-SAIS Markov process described in Section \ref{Sec: MarkovModel} falls in the broad class of generalized epidemic modeling framework (GEMF) introduced in \cite{Sahneh2013TON} for spreading processes on multilayer networks. In essence, the switching contact network of the AC-SAIS model can be equivalently described as a spreading process on multilayer network $\mathcal{G}=(V,E_{S},E_{A})$, where each layer determines the interaction neighborhood that induces state change in a node, depending on its current state. Note that we only need to include layers for susceptible and alert nodes,  because the transition of an infected node towards the susceptible state is spontaneous and does not depend on other nodes states.

Significantly, a multilayer network formulation of adaptive contact reduces the problem from defining a process between $2^{N}$ separate topologies to defining a process on top of a static two--layer network, effectively modeling complex switching dynamics with a conceptually
 straightforward framework. The network layers $G_{S}=(V,E_S)$ and $G_{A}=(V,E_A)$ represent the two extreme cases among all possible $2^{N}$ configurations. The network layer $G_S$ would be physical contact network if none of the nodes were alert, and the network layer $G_A$ would be the physical contact network if none of the nodes were susceptible. Associated with network layers $G_S$ and $G_A$, we define the weighted adjacency matrices $W_{S}=[w_{ij}^S]$ and $W_{A}=[w_{ij}^A]$, respectively. The realized topology at a given time will be a mixture of the two network layers according to the collective node states at that time.  Interestingly though, we show it is
possible to characterize the behavior of the AC-SAIS model in terms of the spectral properties of $W_{S}$ and $W_{A}$ and their interrelation.

{\bf Remark.} The actual, physical/social contact between the network agents is fundamentally different from those represented by the multilayer network $\mathcal{G}$. For example, a directed edge $(i,j)\in E_S$ is physically relevant only if node $i$ is susceptible and node $j$ is infected. Otherwise, node $i$ and $j$ might have a different interaction if, for instance, both are susceptible. However, the later is not relevant for disease spreading and thus no need to be incorporated in our epidemic model. Take for instance the contact between node $i$, who is a nurse, and node $j$, who is a student. These two might not have any social contact in normal situation, however, when node $j$ (the student) is sick, she can possibly pass infection to node $i$ (the nurse); and this is the contact important for epidemic modeling purpose. Also, realize that this contact is directional because when the nurse is sick, she may not have a physical contact with the student. This is why in our state-dependent contact network formulation we do not make ``undirectedness'' assumption on the underlying graph.

\subsection{Mean-Field AC-SAIS Model\label{Sec: Approx}}

Similar to the networked SIS model described in Section \ref{Sec: SIS}, the collective state $\boldmath{X}(t)$ in the AC-SAIS model is a Markov process. However, this Markov process is both
analytically and numerically intractable due to its exponential state space
size of $3^{N}$ (each node can be in one of three states). We can leverage the observation that the AC-SAIS model falls in the GEMF class of stochastic spreading processes on multilayer networks---for which Sahneh et al. \cite{Sahneh2013TON} have derived a system of nonlinear differential equations describing the evolution of state-occupancy probabilities after adopting a first-order, mean-field-type approximation. 



Following procedures explained in \cite{Sahneh2013TON}, we find the first order mean-field-type approximate model for the AC-SAIS model as
\begin{align}
\dot{p}_{i}  &  =-\delta p_{i}+\beta(1-q_{i}-p_{i})\sum w_{ij}^{S}p_{j}+\beta
q_{i}\sum w_{ij}^{A}p_{j},\label{pi_dot}\\
\dot{q}_{i}  &  =\kappa(1-q_{i}-p_{i})\sum w_{ij}^{S}p_{j}-\beta q_{i}\sum
w_{ij}^{A}p_{j}, \label{qi_dot}%
\end{align}
for $i\in\{1,...,N\}$, where $p_{i}$ corresponds to the probability that individual $i$ is
infected, and $q_{i}$ corresponds to the probability that she is alert.

It is worthwhile to acknowledge the limitations of mean-field models. Statistical physics tells us that MF approximations function suitably for infinite-dimensional networks. While, they can perform very poorly for sparse or highly structured networks, such as rings or low-dimensional lattices, particularly close to critical model parameters. Despite, the approximation allows for investigating extremely complex dynamics, and discovering intriguing phenomena and key network characteristics influencing them.

\section{Analysis of AC-SAIS Model\label{Sec: ETE}}
In this section, we compute and study the epidemic
threshold of the mean-field AC-SAIS model in Eqs. (\ref{pi_dot}--\ref{qi_dot}) through analyzing its equilibrium points. Our motivation for this approach stems from the mean-field SIS model which exhibits a threshold phenomena in its equilibrium where a stable (see, \cite{khanafer2014stability,khanafer2014stability2,bonaccorsi2015epidemic}) endemic equilibrium emerges \cite{van2009TN}.

To facilitate the subsequent analysis, we make the following assumption on the structure of the default and adapted neighborhoods throughout this article.
\begin{condition}
	\label{sc}The edge sets $E_{S}$ and $E_{A}$ are such that the two-layer network $\mathcal{G}=(V,E_{S},E_{A})$ is M--connected according to Definition \ref{def: sc}. 
\end{condition}

\subsection{Mean-Field Epidemic Threshold Equation}
 Our approach to finding the critical value $\tau_c$ for AC-SAIS model (Eqs. \ref{pi_dot} and \ref{qi_dot}) is through examining the equilibrium points; as used by Van Mieghem for the SIS model in \cite{van2009TN}. The idea is to show that for $\tau>\tau_c$ an endemic equilibrium ($p^*_i>0,\, \forall i)$ exists aside from the disease-free equilibrium\footnote{Note that $\tau_c$ is the mean-field model threshold value which is a lower bound of the actual value in the exact stochastic SIS process.}. In this approach, strong connectivity of the underlying contact network is pivotal. In case of the SIS model, Van Mieghem \cite{van2009TN} showed that if the contact graph is strongly connected, equilibriums of the mean-field model must either be all zero---the disease-free equilibrium---or they must be strictly positive---the endemic equilibrium. Following lemma shows that similar argument holds for the AC-SAIS model (Eqs. \ref{pi_dot},\ref{qi_dot}) under the M--connectivity assumption of the multilayer network $\mathcal{G}$ as in Definition \ref{def: sc}.

\begin{lemma}
Under Assumption \ref{sc}, the equilibrium value of the infection
probability $p_{i}^{\ast}$ is either zero for all individuals, or strictly positive for all individuals. Moreover, a positive equilibrium satisfies:
\begin{equation}
\frac{p_{i}^{\ast}}{1-p_{i}^{\ast}}=\tau\{\frac{(\bar{\kappa}+1)\sum
	w_{ij}^{A}p_{j}^{\ast}\sum w_{ij}^{S}p_{j}^{\ast}}{\bar{\kappa}\sum w_{ij}^{S}p_{j}^{\ast}+\sum
	w_{ij}^{A}p_{j}^{\ast}}\}, \label{pi_ss}
\end{equation}
with {\emph effective infection rate} $\tau$ and {\emph relative alerting rate} $\bar{\kappa}$ respectively defined as\footnote{According to Poisson processes theory, the effective infection rate $\tau=\beta/\delta$ is equal to the expected number of attempts per link that an infected node makes to infect her neighbor during her infectious period \cite{van2009performance}. The relative alerting rate $\bar{\kappa}=\kappa/\beta$ indicates the ratio of the chance that an infected neighbor cause her neighbor to become alert versus the causing her to become infected. For instance, $\bar{\kappa}=\frac{1}{2}$ means the chance that a node becomes infected from her infected neighbor is twice the chance of becoming alert as the result of interacting with the same neighbor.}
\[
\tau\triangleq\beta/\delta,\quad
\bar{\kappa}\triangleq\frac{\kappa}{\beta}.
\]
\end{lemma}

\begin{proof}
Assume $p_{j}^{\ast}>0$. Letting $\dot{q}_i=0$ in Eq. (\ref{qi_dot}) for any node $i$ with $w_{ij}^{S}>0$ or
$w_{ij}^{A}>0$, yields
\begin{equation}
q_{i}^{\ast}=\frac{\bar{\kappa}\sum w_{ij}^{S}p_{j}^{\ast}}{\bar{\kappa}\sum
w_{ij}^{S}p_{j}^{\ast}+\sum w_{ij}^{A}p_{j}^{\ast}}(1-p_{i}^{\ast
}),\label{qi_ss}%
\end{equation}
Therefore, according to Eq. (\ref{pi_dot}), the equilibrium infection
probabilities $p_{i}^{\ast}$ satisfy%
\begin{equation}
p_{i}^{\ast}=\beta\frac{(1-q_{i}^{\ast})\sum w_{ij}^{S}p_{j}^{\ast}%
+q_{i}^{\ast}\sum w_{ij}^{A}p_{j}^{\ast}}{\delta+\beta\sum w_{ij}^{S}%
p_{j}^{\ast}}.\label{pi_ss0}%
\end{equation}
Replacing $q^*_i$ from Eq. (\ref{qi_ss}) in Eq. (\ref{pi_ss0}) yields the formula in Eq. (\ref{pi_ss}).

The rest of the proof concerns choosing $i$ deliberately, so that $p_{j}^{\ast}>0$ guarantees $p_{i}^{\ast}>0$, and repeating the process until concluding positive equilibrium probabilities for all nodes. We employ the definition of graphs $\boldsymbol{G}^k=(\mathcal{P}_k,\mathcal{L}_k)$ associated with the multilayer network $\mathcal{G}=(V,E_S,E_A)$ as explained in Section \ref{Sec: Multilayer}. According to Definition \ref{def: sc}, if $\mathcal{G}$ is M--connected, there exists $k_*$ such that $\boldsymbol{G}^{k_*}$ is a strongly connected graph. Eq. (\ref{pi_ss}) indicates that in order to get $p^*_i>0$, both $\sum
w_{ij}^{A}p_{j}^{\ast}$ and $\sum w_{ij}^{S}p_{j}^{\ast}$ must be positive. Therefore, choosing $i$ such that $(i,j)\in \mathcal{L}_1$ (for which $w_{ij}^S>0$ and $w_{ij}^A>0$) necessitates $p_{i}^{\ast}>0$. Repeating this process yields the equilibrium probability of all the nodes in the strongly connected component of $G^{1}$ that contains $j$ are all positive. This strongly connected component of $G^{1}$ becomes a single hypernode, which we call $J\in \mathcal{P}_2$, for graph $\boldsymbol{G}^2$. So far, we have proved that $\forall j\in J,\, p^*_j>0$. According to the definition of $\boldsymbol{G}^k$, for graph $\boldsymbol{G}^2$, there is a directed link from component $J$ to component $I$, i.e., $(I,J)\in\mathcal{L}_2$, if and only if 
\begin{equation}\label{ir}
	\exists i\in I~~\text{such that}~~w^{A}_{ij_{1}}, w^{S}_{ij_{2}}>0 ~~\text{for some}~~j_{1},j_{2}\in J.
\end{equation}
Since $\forall j\in J,\, p^*_j>0$, we get $p^*_i>0$ for the above choice of $i$, which further indicates all the nodes of $I$ have positive equilibrium values. As a result, all the nodes belonging to the strongly connected component of $\boldsymbol{G}^2$ that contains $J$ have positive equilibrium values. This procedure can be repeated for $\boldsymbol{G}^3,...,\boldsymbol{G}^{k_*}$. Since, $\boldsymbol{G}^{k_*}$ is strongly connected, all the nodes of the network must have positive equilibrium values.
\end{proof}

We can find the epidemic threshold by examining the
equilibrium points in Eq. (\ref{pi_ss}). For $\tau<\tau_{c}$, the disease-free state is the only equilibrium. However, for $\tau>\tau_{c}$, another equilibrium point
$\boldsymbol{p}^{\ast}\triangleq\lbrack p_{1}^{\ast},...,p_{N}^{\ast}]^{T}\succ0$,
also exists in the positive orthant. Therefore, we find the threshold value of $\tau_{c}$ if we can find a critical value $\tau=\tau_c$ such that
$p_{i}^{\ast}|_{\tau=\tau_{c}}=0$ while
$\frac{dp_{i}^{\ast}}{d\tau}|_{\tau=\tau_{c}}>0$ for all $i\in\{1,...,N\}$. We have the following theorem regarding the value of the epidemic threshold. We would like to emphasize that such a threshold corresponds to the mean-field approximate model (Eqs. \ref{pi_dot}--\ref{qi_dot} ) and should not be confused as the actual threshold value in the exact AC-SAIS Markov model.

\begin{theorem}
\label{Theorem: NPF}The threshold value $\tau_{c}$ for AC-SAIS model
(\ref{pi_dot}-\ref{qi_dot}) is such that the equation%
\begin{equation}
\boldsymbol{z}=\tau_{c}(\bar{\kappa}+1)F(\boldsymbol{z}),
\label{ThresholdEquation}
\end{equation}
with
\begin{equation}
F(\boldsymbol{z})_{i}\triangleq\frac{\sum w_{ij}^{A}z_{j}\sum w_{ij}^{S}z_{j}}{\bar{\kappa}\sum
	w_{ij}^{S}z_{j}+\sum w_{ij}^{A}z_{j}},
\label{Fdef}
\end{equation}
has a nontrivial solution $\boldsymbol{z}\triangleq\lbrack z_{1},...,z_{N}%
]^{T}\succ0$.
\end{theorem}

\begin{proof}
	Equation (\ref{pi_ss}) can be rewritten as
	\[
	\frac{p_{i}^{\ast}}{\sum w_{ij}^{S}p_{j}^{\ast}}=\tau({1-p_{i}^{\ast}})\{\frac{(\bar{\kappa}+1)\sum
		w_{ij}^{A}p_{j}^{\ast}}{\bar{\kappa}\sum w_{ij}^{S}p_{j}^{\ast}+\sum
		w_{ij}^{A}p_{j}^{\ast}}\}.
\]
Now, we take the limit of both sides as $\tau\downarrow\tau_c$, for which $p_i^*\downarrow 0$ for $\forall i$ according to the definition of an epidemic threshold. Since the limit of numerator and denominator of fraction terms of both sides goes to zero, we apply the L'H\^opital's rule for limits \cite{rudin1964principles}:
	\[
	\lim\limits_{\tau\downarrow\tau_c}\frac{ \frac{d}{d\tau}p_{i}^{\ast}}{{\sum w_{ij}^{S}\frac{d}{d\tau}p_{j}^{\ast}}}=\tau_c\lim\limits_{\tau\downarrow\tau_c}\frac{(\bar{\kappa}+1)\sum
		w_{ij}^{A}\frac{d}{d\tau}p_{j}^{\ast}}{\bar{\kappa}\sum w_{ij}^{S}\frac{d}{d\tau}p_{j}^{\ast}+\sum
		w_{ij}^{A}\frac{d}{d\tau}p_{j}^{\ast}}.
	\]
Defining $z_{i}\triangleq\frac{d}{d\tau}p_{i}^{\ast}|_{\tau
	=\tau_{c}}$, the above equation will lead to (\ref{ThresholdEquation}).	The value of $\tau_{c}$ that solves (\ref{ThresholdEquation}) is the critical
value for which $p_{i}^{\ast}=0$, however, $dp_{i}^{\ast}/d\tau>0$, denoting a
second-order phase transition at $\tau=\tau_{c}$. Therefore, $\tau_{c}$ is the
epidemic threshold for AC-SAIS model (\ref{pi_dot}-\ref{qi_dot}).
\end{proof}

Letting $\bar{\kappa}=0$ in Eq. (\ref{Fdef}) yields $F(\boldsymbol{z})=W_S\boldsymbol{z}$, which reduces Eq. (\ref{ThresholdEquation}) to the Perron Frobenius problem $\boldsymbol{z}=\tau_{c}W_{S}\boldsymbol{z}$, suggesting $\tau_{c}=1/\lambda_{1}(W_{S})$; the SIS mean-field threshold. For the AC-SAIS model, the epidemic threshold condition pertains to the
nonlinear Perron-Frobenius problem (\ref{ThresholdEquation}). Though an analytical solution is not expected, we can employ the tools of Section \ref{Sec: NLPF}.

In order to employ Theorem \ref{Theorem: NPF}, we should prove our nonlinear map $F$ in Eq. (\ref{Fdef}) is homogeneous and concave, and it satisfies condition \textbf{C2} defined in Definition \ref{C2}. The map $F$ in Eq. (\ref{Fdef}) is defined for interior of the nonnegative cone. We extend the definition to the boundary of the nonnegative cone by letting $F(\boldsymbol{z})_i=0$ whenever $\sum w_{ij}^{S}z_{j}=0$ and $\sum w_{ij}^{A}z_{j}=0$. In this way, $F(\boldsymbol{z})$ is well defined for all $\boldsymbol{z}\succeq0$. It is obvious that $F$ in Eq. (\ref{Fdef}) is a homogeneous map. Concavity of $F$ can be also deduced from the concavity of the function $g:\mathbb{R}^{2}_{+}\rightarrow \mathbb{R}_{+} $ defined as $g(u,v)=\frac{uv}{u+v}$ (which is half of the harmonic average) because the arguments of $u$ and $v$ are linear transformation of $z_i$'s. Next lemma proves that it also satisfies condition \textbf{C2}.

\begin{lemma}
	Function $F$, defined in Eq. (\ref{Fdef}), satisfies condition \textbf{C2} if and only if the multilayer graph $\mathcal{G}=(V,E_S,E_A)$ is M--connected.
\end{lemma}

\begin{proof}
We just argued that $F$, as defined in Eq. (\ref{Fdef}), is homogeneous and concave. Also, for any set $J$, $F_i(e_J)>0$ if and only if there exists $i\notin J$ and $j_1,j_2\in J$ for which $w_{i,j_1}^S>0$ and $w_{i,j_2}^A>0$, i.e., $(i,j_1)\in E_S$ and $(i,j_2)\in E_A$. Therefore, Theorem \ref{Th: MC2} is applicable and proves the lemma.
\end{proof}

Since we showed $F$ in Eq. (\ref{Fdef}) is homogeneous and concave, and satisfies condition \textbf{C2}, we can apply Theorem \ref{peron} to prove existence and uniqueness of a strictly positive solution for $\boldsymbol{z}$ to the nonlinear Perron--Frobenius problem (\ref{ThresholdEquation}).

\begin{corollary}\label{Corollary: AC_SAIS NPF}
If the multilayer graph $\mathcal{G}=(V,E_S,E_A)$ is M--connected, the nonlinear Perron--Frobenius problem (\ref{ThresholdEquation}) has a unique solution $\boldsymbol{z}=\boldsymbol{z}_*\succ 0$ with $||\boldsymbol{z}_*||_2=1$. Furthermore, the following numerical update law will converge asymptotically to $\boldsymbol{z}_*$:
\begin{equation}
\boldsymbol{z}_{k+1}\triangleq\frac{F(\boldsymbol{z}_{k})+c\boldsymbol{z}_{k}}{\left\Vert F(\boldsymbol{z}_{k})+c\boldsymbol{z}_{k}\right\Vert _{2}},\label{Wk_Law}%
\end{equation}
with $c>0$, and the initial state $\boldsymbol{z}_{0}\succ 0$ and $||\boldsymbol{z}_{0}||_2=1$. Moreover, the threshold value is $\tau_{c}=\frac{1}{(\bar{\kappa}+1)(\boldsymbol{z}_*^{T}F(\boldsymbol{z}_*)-c)}$.	
\end{corollary}

\subsection{Possible Solutions to MF Epidemic Threshold\label{sec: approx}}
Corollary \ref{Corollary: AC_SAIS NPF} proves the existence and uniqueness of a solution for the AC-SAIS threshold formula in Eq. (\ref{ThresholdEquation}). Furthermore, the update law of Eq. (\ref{Wk_Law}) suggests a numerical algorithm for finding the threshold value. Interestingly, a numerical experiment in the next section (see, Fig. \ref{tc}) shows that the epidemic threshold value is a non-monotone function of contact adaptation rate (quantified by $\bar{\kappa}$); indicating faster contact adaptation is not necessarily always better in suppressing epidemics. Here, we aim to predict such scenarios without numerically solving the nonlinear Perron-Frobenius problem for the epidemic threshold.

 The idea is perturbing the threshold equation (\ref{ThresholdEquation}) around two extreme cases of $\bar{\kappa}=0$ and $\bar{\kappa}\rightarrow\infty$, for which we know the exact solutions. Specifically, 1) for $\bar{\kappa}=0$, the epidemic threshold is $\tau_{c}|_{\bar{\kappa}=0}=\frac{1}{\lambda_{1}(W_{S})}$ and $\boldsymbol{z}|_{\bar{\kappa}=0}=\boldsymbol{v}_{S}$ is a solution,
where $\boldsymbol{v}_{S}$ is the dominant eigenvector of matrix $W_{S}$; and 2) for $\bar{\kappa}\rightarrow\infty$, the epidemic threshold is $\tau_{c}|_{\bar{\kappa}\rightarrow\infty}=\frac{1}{\lambda_{1}(W_{A})}$ and $\boldsymbol{z}|_{\bar{\kappa}\rightarrow\infty}=\boldsymbol{v}_{A}$ is a solution, where $\boldsymbol{v}_{A}$ is the dominant eigenvector of matrix $W_{A}$. Thus, employing spectral perturbation techniques, we can approximate the threshold value for
small and large values of relative alerting rate $\bar{\kappa}$.

\begin{theorem}
\label{TheoremAprox}The value of the epidemic threshold solving
(\ref{ThresholdEquation}) has the forms%
\begin{equation}
\label{tawc1}
\tau_{c}(\bar{\kappa})  =\frac{1}{\lambda_{1}(W_{S})}(1+\bar{\kappa}(\Psi(W_{S}%
,W_{A})-1))+o(\bar{\kappa}),
\end{equation}
suitable for small values of $\bar{\kappa}$, and
\begin{equation}
\label{tawc2}
\tau_{c}(\bar{\kappa})=\frac{1}{\lambda_{1}(W_{A})}(1+\bar{\kappa}^{-1}(\Psi(W_{A}%
,W_{S})-1))+o(\bar{\kappa}^{-1})
\end{equation}
suitable for large values of $\bar{\kappa}$, where $\Psi(A,B)$ is%
\begin{equation}
\label{PHIAB}
\Psi(A,B)\triangleq\sum\nolimits_{i=1}^{N}u_{i}v_{i}\frac{\sum\nolimits_{j=1}%
^{N}a_{ij}v_{j}}{\sum\nolimits_{j=1}^{N}b_{ij}v_{j}},
\end{equation}
and $\boldsymbol{v}_A=[v_{1},...,v_{N}]^{T}$ and $\boldsymbol{u}_A=[u_{1}%
,...,u_{N}]^{T}$ are the right and left dominant eigenvectors of $A$
corresponding to $\lambda_{1}(A)$ with $\boldsymbol{v}_A^T\boldsymbol{u}_A=1$.
\end{theorem}

Fortunately, spectral perturbation of the nonlinear Perron-Frobenius problem (Eq. \ref{ThresholdEquation}) leads to analytically tractable formulas expressed in terms of spectral properties of individual layers $G_S$ and $G_A$, and their interrelation (as manifested by $\Psi$ terms in Eq. \ref{tawc1} and Eq. \ref{tawc2}). Using Eq. \ref{tawc1} for small values of $\bar{\kappa}$ and  Eq. \ref{tawc2} for large values of $\bar{\kappa}$, we can categorize several solution possibilities for the full range of $\bar{\kappa}$ values.

To reflect more realistic scenarios, we impose the constraint
$\lambda_{1}(W_{S})>\lambda_{1}(W_{A})$, that is, we assume that if all healthy individuals adopted their alert neighborhood simultaneously, they would collectively raise the epidemic threshold value, making their network more robust against epidemics than the default contact graph $G_{S}$.

The three scenarios for the dependency of the threshold value on contact adaptation rate---as shown in Fig. \ref{tc}---can be characterize as the following:
\begin{enumerate}
	\item \textbf{Monotone scenario (}the faster, the better\textbf{):} This is the simplest case where the value of the epidemic threshold increases monotonically with $\bar{\kappa}$, as simulated in Section \ref{Sec: Numerical Sim} and shown in Fig. \ref{tc} by the black curve. The monotone behavior happens if $\Psi(W_{S},W_{A})>1$ and $\Psi(W_{A},W_{S})<1$. Such monotonically increasing curve occurs, for instance, in contact-avoidance cases\footnote{This is because having $w_{ij}^A\leq w_{ij}^S$ yields  $\sum\nolimits_{j=1}^{N}w_{ij}^{S}v_{Sj}>\sum\nolimits_{j=1}^{N} w_{ij}^{A}v_{Sj}$ and $\sum\nolimits_{j=1}^{N}w_{ij}^{S}v_{Aj}>\sum
	\nolimits_{j=1}^{N}w_{ij}^{A}v_{Aj}$, which according to Eq. (\ref{PHIAB}), leads to $\Psi(W_{S},W_{A})>1$ and $\Psi(W_{A},W_{S})<1$.} where $w_{ij}^A\leq w_{ij}^S$. In other words, if individuals only reduce contact with their neighbors upon becoming alert, the higher the rate they do so, the better; because the epidemic threshold increases with the alerting rate in this scenario.
	\item \textbf{Overshooting scenario (}moderate even better than fast\textbf{):}  It is possible that an optimal alerting rate $\bar{\kappa}$ exists for which the adaptive network is most robust with respect to spreading infection. In other words, having a moderate contact adaptation rate is even better than than the case where the alerting rate is so large that alerting processes is almost instantaneous. The blue curve in Fig. \ref{tc} corresponds to this case. This scenario happens if $\Psi(W_{A},W_{S})>1$.
	\item \textbf{Undershooting scenario (}adaptation goes wrong if slow\textbf{):} An interesting and important scenario is when $\Psi(W_{S},W_{A})<1$. In this case, the value of the epidemic threshold initially decreases as the value of $\bar{\kappa}$ increases. If the switching rate is not fast enough, the alerting process can unintendedly worsen the infection spreading compared to keeping the default contacts! The red curve in Fig. \ref{tc} depicts such scenario.
\end{enumerate}

The following lemma shows that asymmetry of contacts is critical for observing the latter scenario. 

\begin{lemma}
\label{Theorem: Adverse}If $W_{S}$ and $W_{A}$ are both symmetric, $\Psi
(W_{S},W_{A})$ is lower-bounded as%
\begin{equation}
\Psi(W_{S},W_{A})\geq\frac{\lambda_{1}(W_{S})}{\lambda_{1}(W_{A})}.
\label{psiLB}
\end{equation}

\end{lemma}
Given $\lambda_{1}(W_{S})>\lambda_{1}(W_{A})$ (the alert layer is more robust than the susceptible layer), the right hand side of Eq. (\ref{psiLB}) will be always greater than 1. Hence, for undirected network layers, it is impossible for the critical threshold of the adaptive contact network to go below the critical threshold of the default contacts layer, $G_{S}$. We can conclude that asymmetry of contacts is in part responsible for this unexpected behavior.

\section{Numerical Experiments\label{Sec: Numerical Sim}}

In this section, we perform a numerical study to evaluate our findings. 
For $E_{S}$ edges, we consider the well-known ``Football'' network from \cite{Girvan2002} with $N=115$ nodes and $|E_S|=615$ edges, and spectral radius $\lambda_{1}(W_{S})=10.8$. 
Given $G_S$, we synthesize three adapted contact layers $G_{A1}$, $G_{A2}$, 
and $G_{A3}$ as described bellow, and compute their corresponding threshold values as a function of the relative alerting rate as shown in Fig. \ref{tc}.
\begin{itemize}
\item The spectral radii of $G_{Ai}$ graphs are all equal to $\frac{2}{3}$ of
the spectral radius of $G_{S}$, i.e. $\lambda_{1}(W_{Ai}) =  \frac{2}{3}
\lambda_{1}(W_{S})$. In this way, we ensure that the adapted contacts layers are more robust to epidemic 
spreading compared to the default contacts layer. This can be verified in Fig. \ref{tc} where $\tau_c(\infty)=\frac{3}{2}\tau_c(0)$. Note that $\tau_c(0)$ is the threshold value when $\bar\kappa=0$, i.e., no adaptation occurs, and $\tau_c(\infty)$ corresponds to $\bar\kappa=\infty$ where the contact adaptation occurs instantaneously.

\item For $G_{A1}$, $\Psi(W_{S},W_{A1}) < 1$. From Eq. (\ref{tawc1}), we can predict
 that for small values of $\bar{\kappa}$, the epidemic threshold decreases 
 below $\tau_c(0)$, the threshold if no contact adaptation was in place at all . Therefore, we expect an 
 undershoot in $\tau_{c}(\bar{\kappa})$ as a function of $\bar{\kappa}$. This is the configuration where contact adaptation
 can ``go wrong"; despite the fact that the alert contact network is more 
 robust, switching to it can adversely aid in the spread of infection. The red curve in Fig. \ref{tc} corresponds to this scenario.

\item For $G_{A2}$, $\Psi(W_{A2},W_{S})>1$. From Eq. (\ref{tawc2}), we can predict it is possible to get $\tau_c(\bar{\kappa})>\tau_c(\infty)$, an thus there is a value for $\bar{\kappa}$ for which $\tau_{c}(\bar{\kappa})$ 
is maximum. This is in contrast to $G_{A1}$ in that the epidemic threshold 
for the multilayer network is greater than its constituent layers. In this 
configuration, the characteristics are such that an enhanced robustness is 
created synergistically. The blue curve in Fig. \ref{tc} corresponds to this scenario.

\item Graph $G_{A3}$ is made by decreasing the link weights from $G_{S}$, representing a social-avoidance scenario. As discussed in Section \ref{sec: approx}, we expect to see a monotonic 
increase in the epidemic threshold as the contact adaptation rate increases. The black curve in Fig. \ref{tc} corresponds to this scenario.

\end{itemize}

In order to synthesize $G_{A1}$ and $G_{A2}$, we performed a greedy search to obtain desired values of $\Psi$ functions. For each alert contact graph, $G_{Ai}$, and subsequent multilayer network representation, $\mathcal{G}_i=(V,E_{S},E_{Ai})$, we examine spreading behavior at three effective infection rates $\tau_1=0.9\tau_c(0)$, $\tau_2=1.3\tau_c(0)$, and $\tau_3=1.7\tau_c(0)$, as seen in Fig. \ref{tc} (dotted lines). In our numerical simulations, we have set $\delta=1$, which without loss of generality, chooses the unit of time equal to the expected curing period.  Steady-state solutions to the mean-field AC-SAIS Eqns. \ref{pi_dot} \& \ref{qi_dot} are calculated for $10^{-2} \leq \bar{\kappa} \leq 10^2$ and fraction of population infected $\bar{p}=\frac{1}{N}\sum_{i=1}^{N}p_i$---as the indicator of severity of the spreading---is plotted as a function of the alerting rate in Figures \ref{ga1kappa}--\ref{ga3kappa}.

\begin{figure}[ht]
\includegraphics[width=3in]{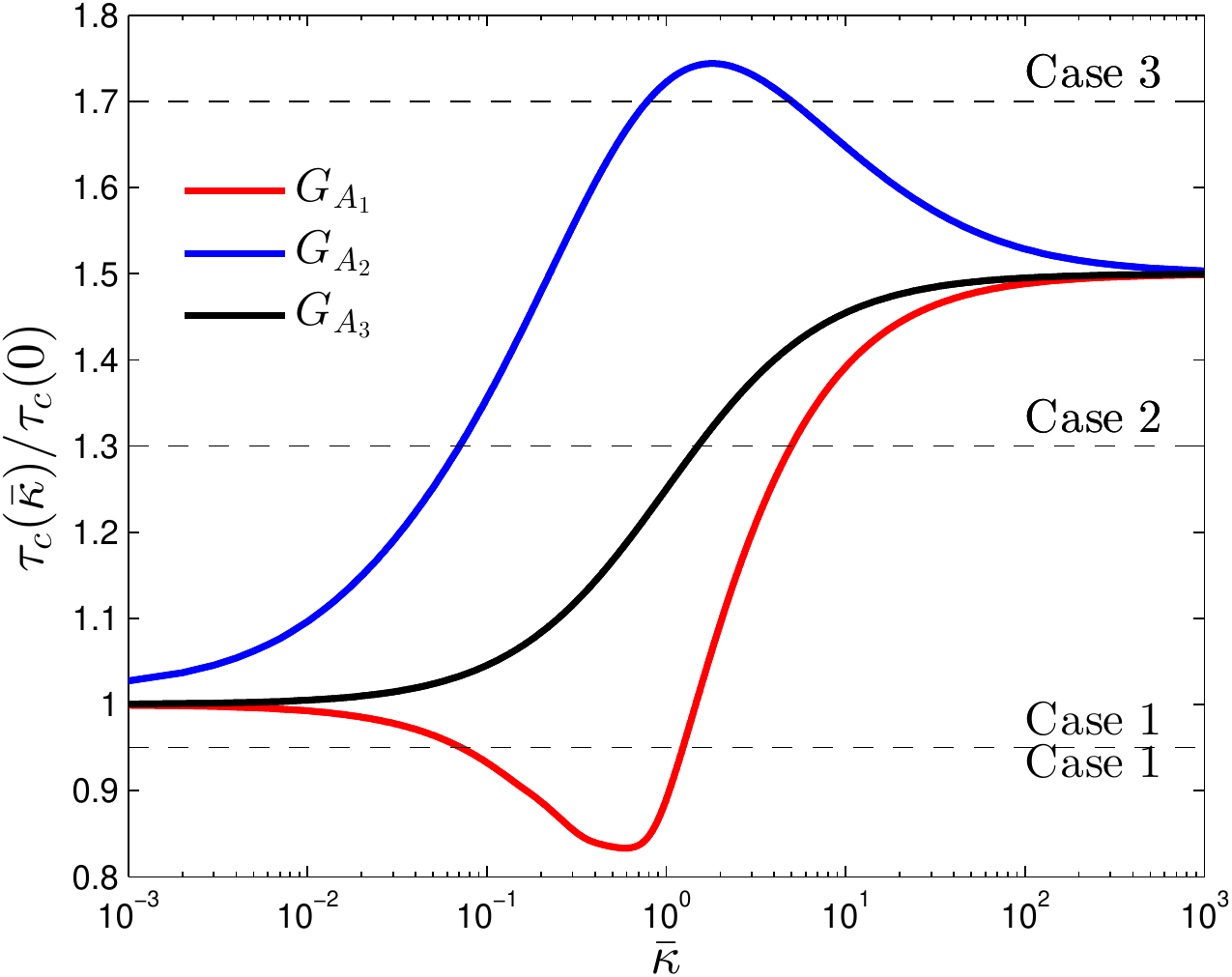}
\caption{Normalized epidemic threshold $\tau_{c}(\bar{\kappa})/\tau_c(0)$ as a function of relative alerting rate $\bar{\kappa}=\frac{\kappa}{\beta}$, showing three dependency scenarios. All three alert layers have the same spectral radius with respect to $G_{S}$ i.e. $\lambda_{1}(W_{S})/\lambda_{1}(W_{Ai})=1.5$. Therefore, in all of them the threshold value $\tau_c(\bar{\kappa})$ starts from $\tau_c(0)=1/\lambda_1(W_S)$ and converges to $\tau_c(\infty)=1.5\tau_c(0)$. Graph $G_{A1}$ is synthesized such that $\Psi(W_S,W_{A1})<1$. From the red curve we can observe that $\tau_{c}(\bar{\kappa})$ decreases for small $\bar{\kappa}$ values after which it increases. Graph $G_{A2}$ is synthesized such that $\Psi(W_{A2},W_S)>1$. In this case the blue curve $\tau_{c}(\bar{\kappa})$ is maximal around $\bar{\kappa}\approx2$. The topology of graph $G_{A3}$ is $G_{S}$ with reduced weights and is represented by the black epidemic threshold curve which increases monotonically by $\bar{\kappa}$.}
\label{tc}
\end{figure}

\subsection{Adaptation gone wrong\label{sec: wrong}}
For the multilayer network with $G_{A1}$ as the adaptive contact layer, we expect to observe increased epidemic sizes --- due to a decreased threshold (red curve of Fig. \ref{tc})--- for a range of low alerting rates.

In \textbf{Case 1}, the infection rate is chosen so that $\tau<\tau_c(0)<\tau_c(\infty)$. In the top plot of Fig. \ref{ga1kappa}, we can see that for most $\bar{\kappa}$ values there is no outbreak, as one would expect since the effective infection rate is below the either extreme values. However, for $0.1 \leq \bar{\kappa} \leq 1.2$ an epidemic is sustained due entirely to inter-layer dynamics creating conditions where an epidemic is more effectively carried throughout the population. In the context of persons altering who they come into contact with, although in an effort to avoid becoming infected, may in fact unintentionally contribute to the opposite outcome.

For \textbf{Case 2}, with $\tau_c(0)<\tau<\tau_c(\infty)$, we observe two regimes of behavior as depicted in the middle plot of Fig. \ref{ga1kappa}: for lower alerting rates, where the effective infection rate is above the epidemic threshold $\tau_c(\bar{\kappa})$, an infection is sustained. For higher alerting rates the reverse is true since the critical threshold goes above $\tau$.

In \textbf{Case 3}, effective infection rate is set above the critical threshold (red curve of Fig. \ref{tc}) for all values of $\bar{\kappa}$, i.e., $\tau_c(\bar{\kappa})<\tau$. Therefore, persistent infections are observed regardless of contact adaptation rate in bottom plot of Fig. \ref{ga1kappa}.

\begin{figure}[ptb]
	\includegraphics[scale=.44,trim={15mm 0 0mm 0mm},clip]{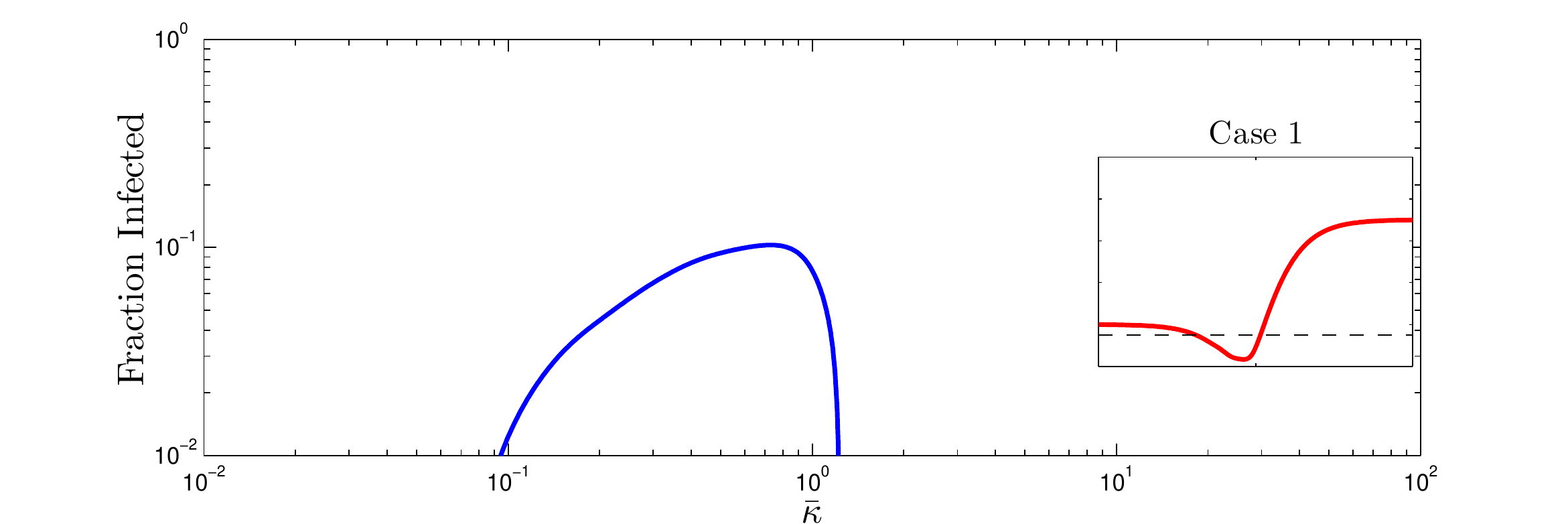} 
	\includegraphics[scale=.44,trim={15mm 0 0mm 0mm},clip]{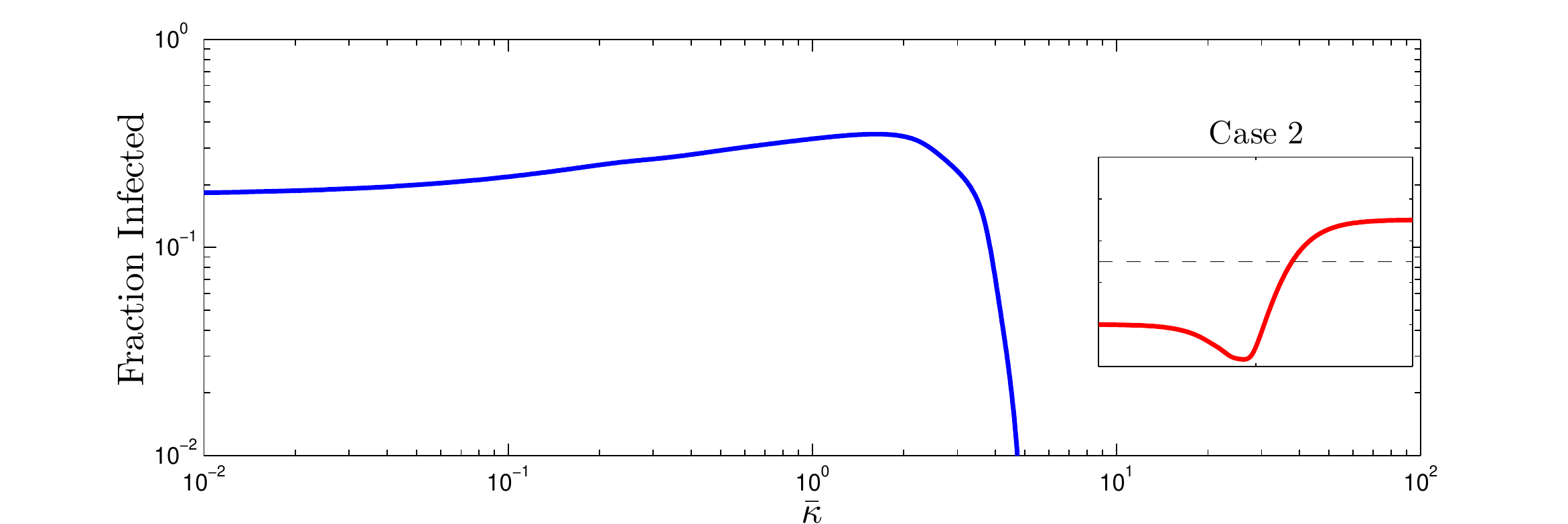} 
	\includegraphics[scale=.44,trim={15mm 0 0mm 0mm},clip]{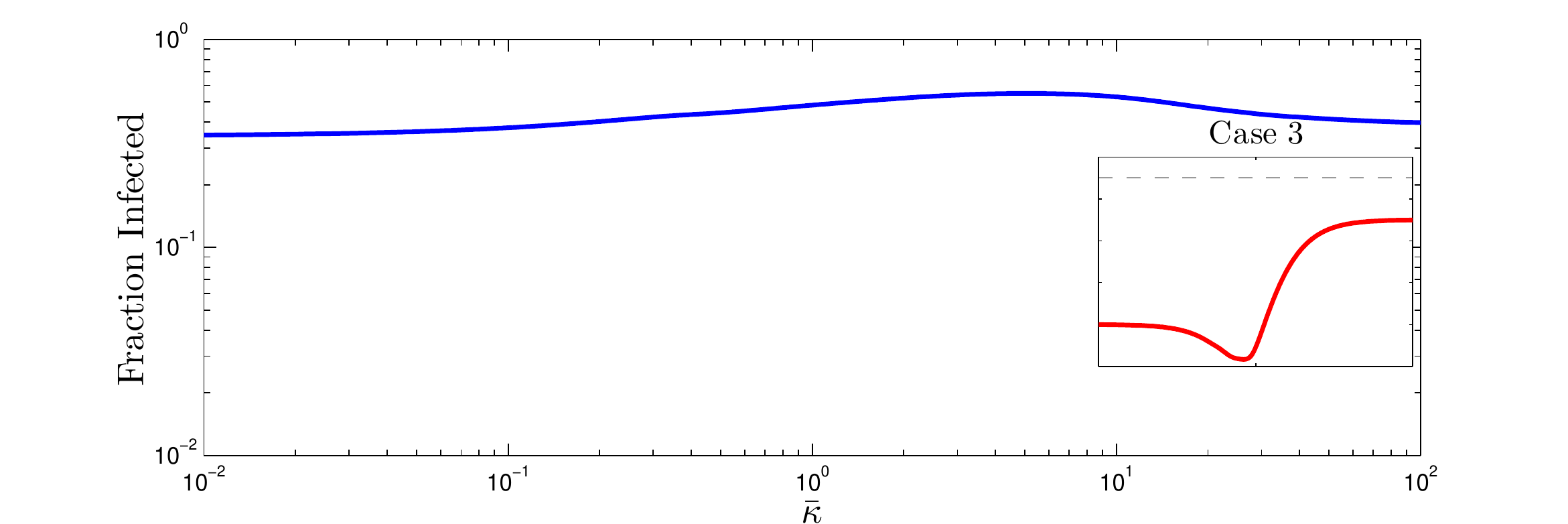} 
	\caption{The effect of alerting rate on infection size for the
		undershooting scenario, for which the epidemic threshold dependence on $\bar\kappa$ is depicted by the red curve in Fig. \ref{tc}. \newline
		\textbf{Case 1} (top) Despite setting the effective infection rate below that of the extreme cases, i.e., $\tau<\tau_c(0)<\tau_c(\infty)$, an epidemic outbreak is still observed for small alerting rates because $\tau$ is larger than the minimum of $\tau_c(\bar{\kappa})$.\newline
		\textbf{case 2} (middle) Effective infection rate lies in between the two extreme values, i.e., $\tau_c(0)<\tau<\tau_c(\infty)$. There is a slight increase in infected individuals after which the infection size drops to $0$ due to the increase in the critical threshold.\newline
		\textbf{Case 3} (bottom) Persistent infections are observed regardless of contact adaptation rate because $\tau_c(\bar{\kappa})<\tau$ for all $\bar{\kappa}$.}
	\label{ga1kappa}
\end{figure}

\subsection{Enhanced robustness\label{sec: enhance}}
We perform the same computations on when the adapted contact layer is $G_{A_2}$. \textbf{Case 1} yields trivially zero infection size. For \textbf{case 2}, shown in the top plot of Fig. \ref{ga2kappa}, we observe that increasing alerting rate beyond a certain value successfully suppresses the infection spreading. \textbf{Case 3}, shown in the bottom plot of Fig. \ref{ga2kappa}, provides an interesting observation in that the critical threshold raises even larger than the alert contacts layer, indicating that a moderate rate of contact adaptation is indeed better than fast rates in enhancing the robustness of the network. Therefore, for $0.8<\bar{\kappa}<5$, the critical threshold increases such that no infection is sustained. While for larger values an outbreak occurs, and the infection size increases as contact adaptation rate increases.

\begin{figure}[htb]
\includegraphics[scale=.44,trim={15mm 0 0mm 0mm},clip]{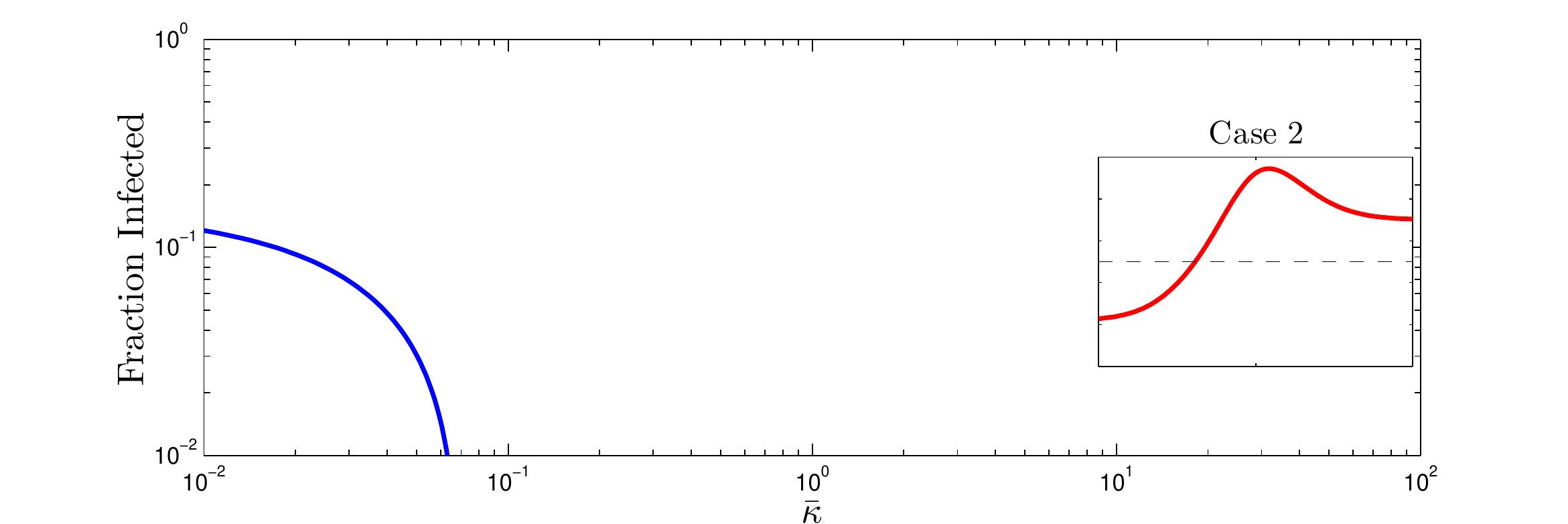} 
\includegraphics[scale=.44,trim={15mm 0 0mm 0mm},clip]{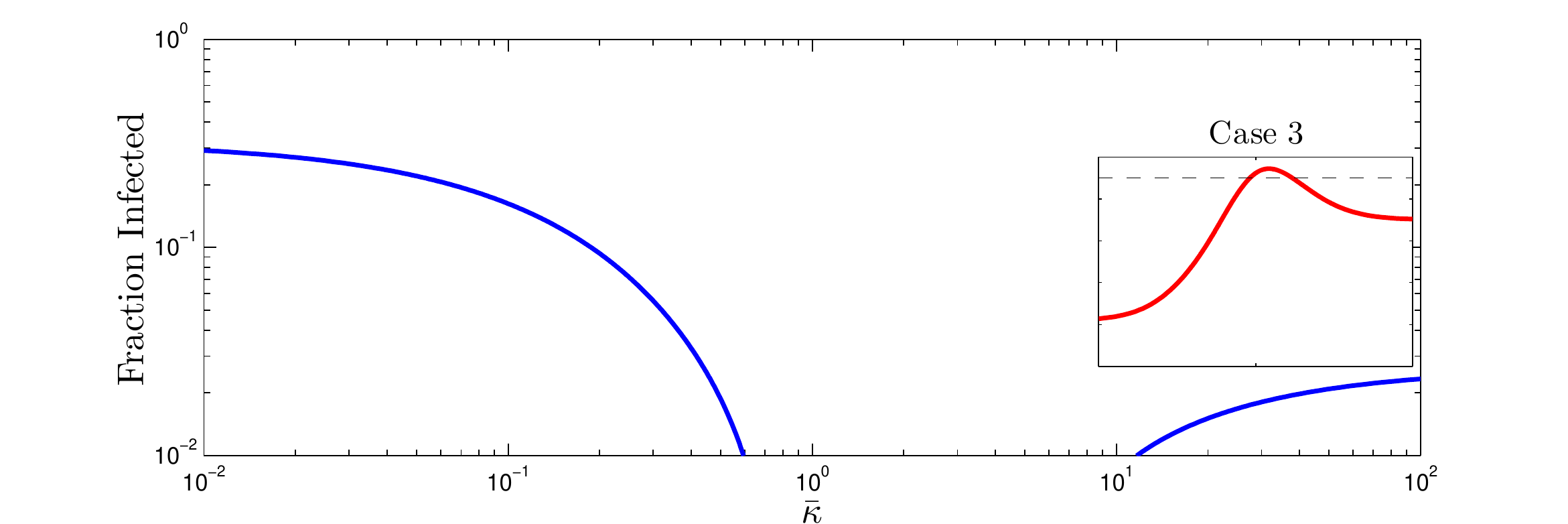} 
 \caption{The effect of alerting rate on infection size for the overshooting scenario, for which the epidemic threshold dependence on $\bar\kappa$ is depicted by the blue curve in Fig. \ref{tc}. \newline
 \textbf{Case 1} This case is omitted since the infection size would be $0$ regardless of the alerting rate\newline
 \textbf{Case 2} (top) The behavior is similar to case 2 with $G_{A1}$ (middle graph in Fig. \ref{ga1kappa}) though the transition to zero infection size occurs at a smaller alerting rate.\newline
 \textbf{Case 3} (bottom) This is a scenario when the effective infection rate is larger than the extreme values $(\tau_c(0)<\tau_c(\infty)<\tau)$, yet it is less than the maximum of the threshold curve $\tau_c(\bar{\kappa})$ as seen by the blue curve in Fig. \ref{tc}. A non-zero infection size is observed for small alerting rates, eventually $\tau_c(\bar{\kappa})$ raises above $\tau$ so that an epidemic cannot be sustained. As the threshold converges towards $\tau_c(\infty)$, an epidemic can once again persist, and the infection size even increases by the contact adaptation rate.}
\label{ga2kappa}
\end{figure}

\subsection{Monotonic Dependency\label{sec: mono}}
In the third scenario, the adapted contact layer is constructed by lowering the edge weights of $G_S$. This would correspond to a social distancing scenario, where individuals limit or abandon their contacts when they become alert. As can be seen by the black curve in Fig. \ref{tc}, the threshold value increases monotonically by the alerting rate. Fig. \ref{ga3kappa} depicts the second case where $\tau_c(0)<\tau<\tau_c(\infty)$. As expected, there is a certain value of $\bar\kappa*$ so that the epidemic infection is controlled for alerting rates $\bar\kappa\geq\bar\kappa*$. Case 1 and 3 are omitted for trivial behavior.

\begin{figure}[h]
\includegraphics[scale=.44]{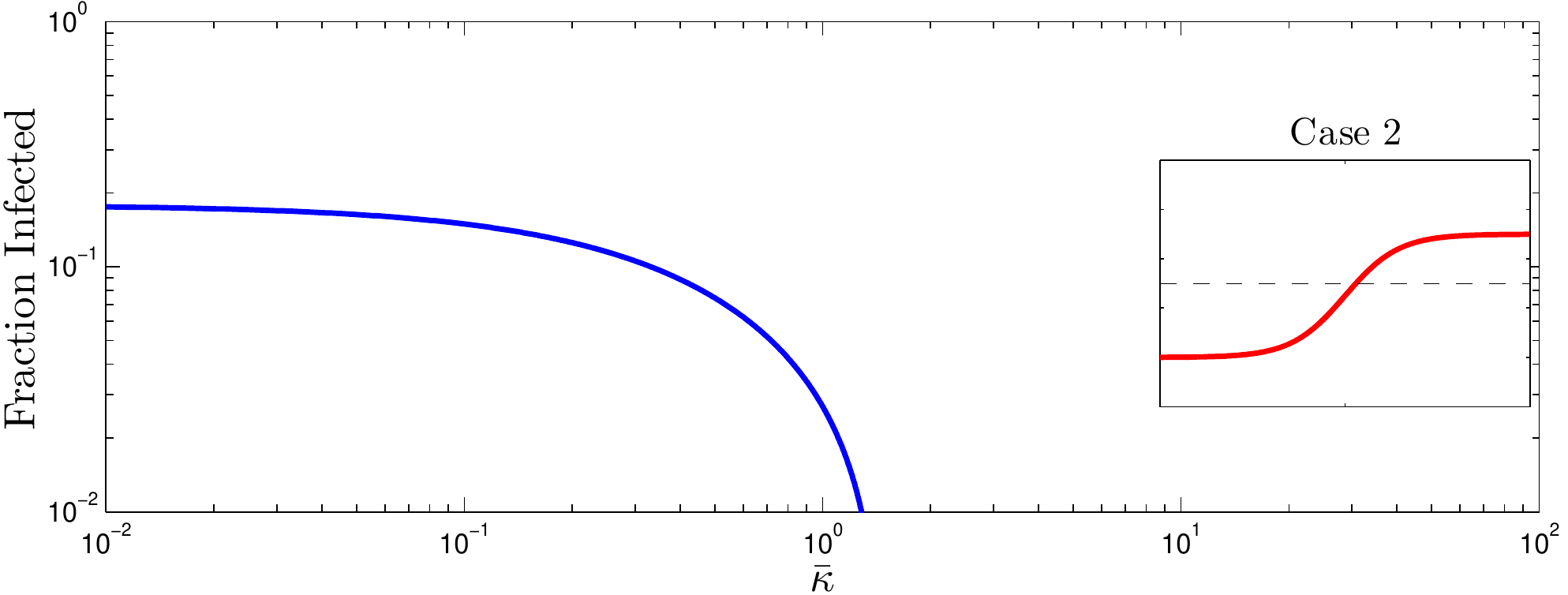} 
 \caption{The effect of alerting rate on infection size for
 the monotone scenario, for which the epidemic threshold dependence on $\bar\kappa$ is depicted by the black curve in Fig. \ref{tc}. \newline
 \textbf{Case 2} Similar to Sections \ref{sec: wrong} and \ref{sec: enhance}, case 2 shows a transition between low and high alerting rates where epidemic outbreaks occur for the former and not the latter.\newline
Cases 1 and 3 are omitted for trivial behavior.}
\label{ga3kappa}
\end{figure}

\section{Discussions and Conclusion}

The state-dependent switching (adaptive) contact network in the AC-SAIS model leads to rich dynamics for the epidemic spreading process and behavior not yet identified in literature (to the best of the authors' knowledge). Intuitively, when nodes can ``switch'' to a neighborhood constituting a more robust network, the expected effect on the overall robustness of the network would be to increase monotonically with the alerting rate. As shown in Section \ref{sec: enhance} and \ref{sec: wrong}, this is not always the case. Indeed, we observed non-monotone dependency of the epidemic threshold in most of our experiment trials. We show how the adaptive switching topology of the contact network is different from fixed static graphs and can lead to regimes of extreme or unexpected behavior. In particular, it is possible that adaptive behavior towards a supposedly more resilient network can in fact worsen the severity of an outbreak, or enable the possibility where it did not exist before. On the other hand, it is possible to configure network layers such that the multilayer network is more robust than either individual layer.


It is noteworthy to mention that some results in the literature point to the observation that contact adaptation do not always aid suppressing the infection. For example, Meloni et al. \cite{Ves2011SciRep} considered a self-initiated behavior where individuals change their mobility patterns. When travelers decide to avoid locations with high levels of infection and travel through locations with low levels of infections, this behavioral change may facilitate disease spreading because individuals effectively act as vectors of the disease transmission. It is very important to highlight the difference of the underlying mechanism between these formerly reported results and the ``adaptation-gone-wrong'' behavior in this paper. In our model, individuals who adapt their contacts (alerts) do not act as vectors for propagating the infection because they do not carry infection. This comes purely as a result of the adaptive behavior, signifying the importance of further analysis of state-dependent networks.

Finally, we would like to highlight several aspects of this study that go beyond the specific epidemic model considered in this paper. We developed a necessary and sufficient condition for existence and uniqueness of a positive eigenvector for homogeneous, concave maps. Furthermore, our utilization of multilayer networks to formulate dynamics on state-dependent switching networks is novel and can facilitate analysis of many networked dynamical systems. In our analysis, we come up with concepts that are genuine and novel to multilayer networks. Specifically, our analysis leads to 1) critical phenomena characterized by a nonlinear Perron Frobenius equation, and 2) the definition of M--connectivity. Our proposed concept of M--connectivity can easily scale to more than two layers. Furthermore, the joint descriptor $\Psi$  in Eq. \ref{PHIAB} emphasizes the importance of joint descriptors when characterizing dynamics over multilayer networks. While network science has flourished in understanding intra-layer network topologies, intrinsic descriptors of inter-layer connectivity of multilayer networks are yet to be investigated.

\section*{Acknowledgment}
\addcontentsline{toc}{section}{Acknowledgment}
We are grateful to anonymous reviewers for their constructive feedback to improve this manuscript. We would like to thank Victor Preciado for inspiring conversations regarding use of the SAIS model in a contact adaptation context. Also, we would like to thank Heman Shakeri and Rad Niazadeh for their helpful discussions on the nonlinear Perron-Frobenius theory.

\bibliography{acsais} 

\begin{thebibliography}{10}
\providecommand{\url}[1]{#1}
\csname url@samestyle\endcsname
\providecommand{\newblock}{\relax}
\providecommand{\bibinfo}[2]{#2}
\providecommand{\BIBentrySTDinterwordspacing}{\spaceskip=0pt\relax}
\providecommand{\BIBentryALTinterwordstretchfactor}{4}
\providecommand{\BIBentryALTinterwordspacing}{\spaceskip=\fontdimen2\font plus
\BIBentryALTinterwordstretchfactor\fontdimen3\font minus
  \fontdimen4\font\relax}
\providecommand{\BIBforeignlanguage}[2]{{%
\expandafter\ifx\csname l@#1\endcsname\relax
\typeout{** WARNING: IEEEtran.bst: No hyphenation pattern has been}%
\typeout{** loaded for the language `#1'. Using the pattern for}%
\typeout{** the default language instead.}%
\else
\language=\csname l@#1\endcsname
\fi
#2}}
\providecommand{\BIBdecl}{\relax}
\BIBdecl

\bibitem{hethcote2000mathematics}
H.~W. Hethcote, ``The mathematics of infectious diseases,'' \emph{SIAM review},
  vol.~42, no.~4, pp. 599--653, 2000.

\bibitem{keeling2008modeling}
M.~J. Keeling and P.~Rohani, \emph{Modeling infectious diseases in humans and
  animals}.\hskip 1em plus 0.5em minus 0.4em\relax Princeton University Press,
  2008.

\bibitem{anderson1992infectious}
R.~M. Anderson, R.~M. May, and B.~Anderson, \emph{Infectious diseases of
  humans: dynamics and control}.\hskip 1em plus 0.5em minus 0.4em\relax Wiley
  Online Library, 1992, vol.~28.

\bibitem{knight2016bridging}
G.~M. Knight, N.~J. Dharan, G.~J. Fox, N.~Stennis, A.~Zwerling, R.~Khurana, and
  D.~W. Dowdy, ``Bridging the gap between evidence and policy for infectious
  diseases: How models can aid public health decision-making,''
  \emph{International journal of infectious diseases}, vol.~42, pp. 17--23,
  2016.

\bibitem{pastor2015epidemic}
R.~Pastor-Satorras, C.~Castellano, P.~Van~Mieghem, and A.~Vespignani,
  ``Epidemic processes in complex networks,'' \emph{Reviews of modern physics},
  vol.~87, no.~3, p. 925, 2015.

\bibitem{moran2016epidemic}
K.~R. Moran, G.~Fairchild, N.~Generous, K.~Hickmann, D.~Osthus, R.~Priedhorsky,
  J.~Hyman, and S.~Y. Del~Valle, ``Epidemic forecasting is messier than weather
  forecasting: The role of human behavior and internet data streams in epidemic
  forecast,'' \emph{Journal of Infectious Diseases}, vol. 214, no. suppl 4, pp.
  S404--S408, 2016.

\bibitem{bish2010demographic}
A.~Bish and S.~Michie, ``Demographic and attitudinal determinants of protective
  behaviours during a pandemic: a review,'' \emph{British journal of health
  psychology}, vol.~15, no.~4, pp. 797--824, 2010.

\bibitem{Book2013HumanBehavior}
P.~Manfredi and A.~d'Onofrio, \emph{Modeling the interplay between human
  behavior and the spread of infectious diseases}.\hskip 1em plus 0.5em minus
  0.4em\relax Springer, 2013.

\bibitem{funk2010JRSI}
S.~Funk, M.~Salath{\'e}, and V.~A. Jansen, ``Modelling the influence of human
  behaviour on the spread of infectious diseases: a review,'' \emph{Journal of
  the Royal Society Interface}, vol.~7, no.~50, pp. 1247--1256, 2010.

\bibitem{verelst2016behavioural}
F.~Verelst, L.~Willem, and P.~Beutels, ``Behavioural change models for
  infectious disease transmission: a systematic review (2010--2015),''
  \emph{Journal of The Royal Society Interface}, vol.~13, no. 125, p. 20160820,
  2016.

\bibitem{wang2015coupled}
Z.~Wang, M.~A. Andrews, Z.-X. Wu, L.~Wang, and C.~T. Bauch, ``Coupled
  disease--behavior dynamics on complex networks: A review,'' \emph{Physics of
  life reviews}, vol.~15, pp. 1--29, 2015.

\bibitem{chen2006JMB}
F.~Chen, ``A susceptible-infected epidemic model with voluntary vaccinations,''
  \emph{Journal of mathematical biology}, vol.~53, no.~2, pp. 253--272, 2006.

\bibitem{funk2010JTB}
S.~Funk, E.~Gilad, and V.~Jansen, ``Endemic disease, awareness, and local
  behavioural response,'' \emph{Journal of Theoretical Biology}, vol. 264,
  no.~2, pp. 501--509, 2010.

\bibitem{funk2009NAS}
S.~Funk, E.~Gilad, C.~Watkins, and V.~Jansen, ``The spread of awareness and its
  impact on epidemic outbreaks,'' \emph{Proceedings of the National Academy of
  Sciences}, vol. 106, no.~16, pp. 6872--6877, 2009.

\bibitem{Nicola2011PLOS}
N.~Perra, D.~Balcan, B.~Gonasalves, and A.~Vespignani, ``Towards a
  characterization of behavior-disease models,'' \emph{PLoS ONE}, vol.~6,
  no.~8, p. e23084, 08 2011.

\bibitem{poletti2009JTB2}
P.~Poletti, B.~Caprile, M.~Ajelli, A.~Pugliese, and S.~Merler, ``Spontaneous
  behavioural changes in response to epidemics,'' \emph{Journal of Theoretical
  Biology}, vol. 260, no.~1, pp. 31--40, 2009.

\bibitem{FaryadCDC11SAIS}
F.~Sahneh and C.~Scoglio, ``Epidemic spread in human networks,'' in
  \emph{Proceedings of IEEE Conference on Decision and Control}, 2011.

\bibitem{reluga2010game}
T.~C. Reluga, ``Game theory of social distancing in response to an epidemic,''
  \emph{PLoS computational biology}, vol.~6, no.~5, p. e1000793, 2010.

\bibitem{Mina2011JTB}
M.~Youssef and C.~Scoglio, ``An individual-based approach to {SIR} epidemics in
  contact networks,'' \emph{Journal of Theoretical Biology}, vol. 283, no.~1,
  pp. 136--144, 2011.

\bibitem{gross2008JRSI}
T.~Gross and B.~Blasius, ``Adaptive coevolutionary networks: a review,''
  \emph{Journal of The Royal Society Interface}, vol.~5, no.~20, pp. 259--271,
  Mar. 2008.

\bibitem{gross2006PRL}
T.~Gross, C.~J.~D. D'Lima, and B.~Blasius, ``Epidemic dynamics on an adaptive
  network,'' \emph{Physical review letters}, vol.~96, no.~20, p. 208701, 2006.

\bibitem{marceau2010PRE}
V.~Marceau, P.~No{\"e}l, L.~H{\'e}bert-Dufresne, A.~Allard, and L.~Dub{\'e},
  ``Adaptive networks: Coevolution of disease and topology,'' \emph{Physical
  Review E}, vol.~82, no.~3, p. 036116, 2010.

\bibitem{risau2009JTB}
S.~Risau-Gusm{\'a}n and D.~Zanette, ``Contact switching as a control strategy
  for epidemic outbreaks,'' \emph{Journal of theoretical biology}, vol. 257,
  no.~1, pp. 52--60, 2009.

\bibitem{demirel2012X}
G.~Demirel and T.~Gross, ``Absence of epidemic thresholds in a growing adaptive
  network,'' \emph{arXiv preprint arXiv:1209.2541}, 2012.

\bibitem{van_segbroeck2010PLoS}
S.~Van~Segbroeck, F.~C. Santos, and J.~M. Pacheco, ``Adaptive contact networks
  change effective disease infectiousness and dynamics,'' \emph{{PLoS} Comput
  Biol}, vol.~6, no.~8, p. e1000895, Aug. 2010.

\bibitem{Ves2011SciRep}
S.~Meloni, N.~Perra, A.~Arenas, S.~Gomez, Y.~Moreno, and A.~Vespignani,
  ``Modeling human mobility responses to the large-scale spreading of
  infectious diseases,'' \emph{Scientific Reports}, vol.~1, no.~62, 2011.

\bibitem{fenichel2011adaptive}
E.~P. Fenichel, C.~Castillo-Chavez, M.~G. Ceddia, G.~Chowell, P.~A.~G. Parra,
  G.~J. Hickling, G.~Holloway, R.~Horan, B.~Morin, C.~Perrings \emph{et~al.},
  ``Adaptive human behavior in epidemiological models,'' \emph{Proceedings of
  the National Academy of Sciences}, vol. 108, no.~15, pp. 6306--6311, 2011.

\bibitem{liu2015endemic}
M.~Liu, E.~Liz, and G.~Röst, ``Endemic bubbles generated by delayed
  behavioral response: global stability and bifurcation switches in an sis
  model,'' \emph{SIAM Journal on Applied Mathematics}, vol.~75, no.~1, pp.
  75--91, 2015.

\bibitem{brauer2011simple}
F.~Brauer, ``A simple model for behaviour change in epidemics,'' \emph{BMC
  public health}, vol.~11, no.~1, p.~S3, 2011.

\bibitem{li2015bifurcation}
J.~Li, Y.~Zhao, and H.~Zhu, ``Bifurcation of an sis model with nonlinear
  contact rate,'' \emph{Journal of Mathematical Analysis and Applications},
  vol. 432, no.~2, pp. 1119--1138, 2015.

\bibitem{morin2014disease}
B.~R. Morin, C.~Perrings, S.~Levin, and A.~Kinzig, ``Disease risk mitigation:
  the equivalence of two selective mixing strategies on aggregate contact
  patterns and resulting epidemic spread,'' \emph{Journal of theoretical
  biology}, vol. 363, pp. 262--270, 2014.

\bibitem{xiao2012sliding}
Y.~Xiao, X.~Xu, and S.~Tang, ``Sliding mode control of outbreaks of emerging
  infectious diseases,'' \emph{Bulletin of mathematical biology}, vol.~74,
  no.~10, pp. 2403--2422, 2012.

\bibitem{paarporn2015epidemic}
K.~Paarporn, C.~Eksin, J.~S. Weitz, and J.~S. Shamma, ``Epidemic spread over
  networks with agent awareness and social distancing,'' in
  \emph{Communication, Control, and Computing (Allerton), 2015 53rd Annual
  Allerton Conference on}.\hskip 1em plus 0.5em minus 0.4em\relax IEEE, 2015,
  pp. 51--57.

\bibitem{sahneh2012SR}
F.~D. Sahneh, F.~N. Chowdhury, and C.~M. Scoglio, ``On the existence of a
  threshold for preventive behavioral responses to suppress epidemic
  spreading,'' \emph{Scientific reports}, vol.~2, 2012.

\bibitem{misra2011effect}
A.~Misra, A.~Sharma, and V.~Singh, ``Effect of awareness programs in
  controlling the prevalence of an epidemic with time delay,'' \emph{Journal of
  Biological Systems}, vol.~19, no.~02, pp. 389--402, 2011.

\bibitem{samanta2013effect}
S.~Samanta, S.~Rana, A.~Sharma, A.~Misra, and J.~Chattopadhyay, ``Effect of
  awareness programs by media on the epidemic outbreaks: a mathematical
  model,'' \emph{Applied Mathematics and Computation}, vol. 219, no.~12, pp.
  6965--6977, 2013.

\bibitem{misra2015stability}
A.~Misra, A.~Sharma, and J.~Shukla, ``Stability analysis and optimal control of
  an epidemic model with awareness programs by media,'' \emph{Biosystems}, vol.
  138, pp. 53--62, 2015.

\bibitem{wang2015interaction}
Q.~Wang, L.~Zhao, R.~Huang, Y.~Yang, and J.~Wu, ``Interaction of media and
  disease dynamics and its impact on emerging infection management,''
  \emph{Discrete and Continuous Dynamical Systems-Series B}, vol.~20, no.~1,
  pp. 215--230, 2015.

\bibitem{van2009TN}
P.~Van~Mieghem, J.~Omic, and R.~Kooij, ``Virus spread in networks,''
  \emph{IEEE/ACM Transactions on Networking}, vol.~17, no.~1, pp. 1--14, 2009.

\bibitem{SahnehCDC12}
F.~Sahneh and C.~Scoglio, ``Optimal information dissemination in epidemic
  networks,'' in \emph{Decision and Control (CDC), 2012 IEEE 51st Annual
  Conference on}, Dec 2012, pp. 1657--1662.

\bibitem{bonacich2007some}
P.~Bonacich, ``Some unique properties of eigenvector centrality,'' \emph{Social
  networks}, vol.~29, no.~4, pp. 555--564, 2007.

\bibitem{Preciado2013SAIS}
V.~Preciado, F.~Sahneh, and C.~Scoglio, ``A convex framework for optimal
  investment on disease awareness in social networks,'' in \emph{Global
  Conference on Signal and Information Processing (GlobalSIP), 2013 IEEE}, Dec
  2013, pp. 851--854.

\bibitem{shakeri2015optimal}
H.~Shakeri, F.~D. Sahneh, C.~Scoglio, P.~Poggi-Corradini, and V.~M. Preciado,
  ``Optimal information dissemination strategy to promote preventive behaviours
  in multilayer epidemic networks,'' \emph{Math. Biosc. Eng}, vol.~12, no.~3,
  pp. 609--623, 2015.

\bibitem{juher2015analysis}
D.~Juher, I.~Z. Kiss, and J.~Salda{\~n}a, ``Analysis of an epidemic model with
  awareness decay on regular random networks,'' \emph{Journal of theoretical
  biology}, vol. 365, pp. 457--468, 2015.

\bibitem{maharaj2012controlling}
S.~Maharaj and A.~Kleczkowski, ``Controlling epidemic spread by social
  distancing: Do it well or not at all,'' \emph{BMC Public Health}, vol.~12,
  no.~1, p. 679, 2012.

\bibitem{Kivela2014}
M.~Kivela, M., A.~Arenas, M.~Barthelemy, J.~P. Gleeson, Y.~Moreno, and M.~A.
  Porter, ``Multilayer networks,'' \emph{Journal of complex networks}, vol.~2,
  no.~3, pp. 203--271, 2014.

\bibitem{valdez2012intermittent}
L.~Valdez, P.~A. Macri, and L.~A. Braunstein, ``Intermittent social distancing
  strategy for epidemic control,'' \emph{Physical Review E}, vol.~85, no.~3, p.
  036108, 2012.

\bibitem{VanMieghem2013PRE}
D.~Guo, S.~Trajanovski, R.~van~de Bovenkamp, H.~Wang, and P.~Van~Mieghem,
  ``Epidemic threshold and topological structure of
  susceptible-infectious-susceptible epidemics in adaptive networks,''
  \emph{Phys. Rev. E}, vol.~88, p. 042802, Oct 2013.

\bibitem{juher2013outbreak}
D.~Juher, J.~Ripoll, and J.~Salda{\~n}a, ``Outbreak analysis of an sis epidemic
  model with rewiring,'' \emph{Journal of mathematical biology}, pp. 1--22,
  2013.

\bibitem{dong2015can}
C.~Dong, Q.~Yin, W.~Liu, Z.~Yan, and T.~Shi, ``Can rewiring strategy control
  the epidemic spreading?'' \emph{Physica A: Statistical Mechanics and its
  Applications}, vol. 438, pp. 169--177, 2015.

\bibitem{Ogura2016}
M.~Ogura and V.~M. Preciado, ``Epidemic processes over adaptive state-dependent
  networks,'' \emph{Physical Review E}, vol.~93, no.~6, p. 062316, 2016.

\bibitem{ogura2015disease}
------, ``Disease spread over randomly switched large-scale networks,'' in
  \emph{2015 American Control Conference (ACC)}.\hskip 1em plus 0.5em minus
  0.4em\relax IEEE, 2015, pp. 1782--1787.

\bibitem{ogura2016stability}
------, ``Stability of spreading processes over time-varying large-scale
  networks,'' \emph{IEEE Transactions on Network Science and Engineering},
  vol.~3, no.~1, pp. 44--57, 2016.

\bibitem{van2010graph}
P.~Van~Mieghem, \emph{Graph spectra for complex networks}.\hskip 1em plus 0.5em
  minus 0.4em\relax Cambridge University Press, 2010.

\bibitem{krasnoselskij1964positive}
M.~A. Krasnoselskij, ``Positive solutions of operator equations,'' 1964.

\bibitem{per1}
S.~Gaubert and J.~Gunawardena, ``The perron-frobenius theorem for homogeneous,
  monotone functions,'' \emph{Transactions of the American Mathematical
  Society}, vol. 356, no.~12, pp. 4931--4950, 2004.

\bibitem{per2}
B.~Lemmens and R.~Nussbaum, \emph{Nonlinear Perron-Frobenius Theory}.\hskip 1em
  plus 0.5em minus 0.4em\relax Cambridge University Press, 2012, vol. 189.

\bibitem{krause2001concave}
U.~Krause, ``Concave perron--frobenius theory and applications,''
  \emph{Nonlinear Analysis: Theory, Methods \& Applications}, vol.~47, no.~3,
  pp. 1457--1466, 2001.

\bibitem{nussbaum1999generalizations}
R.~D. Nussbaum and S.~M.~V. Lunel, \emph{Generalizations of the
  Perron-Frobenius theorem for nonlinear maps}.\hskip 1em plus 0.5em minus
  0.4em\relax American Mathematical Soc., 1999, vol. 659.

\bibitem{wood2004always}
R.~Wood and M.~O'Neill, ``An always convergent method for finding the spectral
  radius of an irreducible non-negative matrix,'' \emph{ANZIAM Journal},
  vol.~45, pp. 474--485, 2004.

\bibitem{boccaletti2014structure}
S.~Boccaletti, G.~Bianconi, R.~Criado, C.~I. Del~Genio, J.~G{\'o}mez-Gardenes,
  M.~Romance, I.~Sendina-Nadal, Z.~Wang, and M.~Zanin, ``The structure and
  dynamics of multilayer networks,'' \emph{Physics Reports}, vol. 544, no.~1,
  pp. 1--122, 2014.

\bibitem{de2013mathematical}
M.~De~Domenico, A.~Sol{\'e}-Ribalta, E.~Cozzo, M.~Kivel{\"a}, Y.~Moreno, M.~A.
  Porter, S.~G{\'o}mez, and A.~Arenas, ``Mathematical formulation of multilayer
  networks,'' \emph{Physical Review X}, vol.~3, no.~4, p. 041022, 2013.

\bibitem{ganesh2005effect}
A.~Ganesh, L.~Massouli{\'e}, and D.~Towsley, ``The effect of network topology
  on the spread of epidemics,'' in \emph{INFOCOM 2005. 24th Annual Joint
  Conference of the IEEE Computer and Communications Societies. Proceedings
  IEEE}, vol.~2.\hskip 1em plus 0.5em minus 0.4em\relax IEEE, 2005, pp.
  1455--1466.

\bibitem{lloyd2001viruses}
A.~L. Lloyd and R.~M. May, ``How viruses spread among computers and people,''
  \emph{Science}, vol. 292, no. 5520, pp. 1316--1317, 2001.

\bibitem{neal2014endemic}
P.~Neal \emph{et~al.}, ``Endemic behaviour of sis epidemics with general
  infectious period distributions,'' \emph{Advances in Applied Probability},
  vol.~46, no.~1, pp. 241--255, 2014.

\bibitem{van2013non}
P.~Van~Mieghem and R.~Van~de Bovenkamp, ``Non-markovian infection spread
  dramatically alters the susceptible-infected-susceptible epidemic threshold
  in networks,'' \emph{Physical review letters}, vol. 110, no.~10, p. 108701,
  2013.

\bibitem{cator2013susceptible}
E.~Cator, R.~Van~de Bovenkamp, and P.~Van~Mieghem,
  ``Susceptible-infected-susceptible epidemics on networks with general
  infection and cure times,'' \emph{Physical Review E}, vol.~87, no.~6, p.
  062816, 2013.

\bibitem{Sahneh2013TON}
F.~D. Sahneh, C.~Scoglio, and P.~Van~Mieghem, ``Generalized epidemic mean-field
  model for spreading processes over multilayer complex networks,''
  \emph{IEEE/ACM Transactions on Networking}, vol.~21, no.~5, pp. 1609--1620,
  2013.

\bibitem{gleeson2012accuracy}
J.~P. Gleeson, S.~Melnik, J.~A. Ward, M.~A. Porter, and P.~J. Mucha, ``Accuracy
  of mean-field theory for dynamics on real-world networks,'' \emph{Physical
  Review E}, vol.~85, no.~2, p. 026106, 2012.

\bibitem{taylor2012markovian}
M.~Taylor, P.~L. Simon, D.~M. Green, T.~House, and I.~Z. Kiss, ``From markovian
  to pairwise epidemic models and the performance of moment closure
  approximations,'' \emph{Journal of mathematical biology}, vol.~64, no.~6, pp.
  1021--1042, 2012.

\bibitem{Pastor-Satorras2015}
R.~Pastor-Satorras, C.~Castellano, P.~Van~Mieghem, and A.~Vespignani,
  ``Epidemic processes in complex networks,'' \emph{Reviews of modern physics},
  vol.~87, no.~3, p. 925, 2015.

\bibitem{khanafer2014stability}
A.~Khanafer, T.~Basar, and B.~Gharesifard, ``Stability properties of infected
  networks with low curing rates,'' in \emph{American Control Conference (ACC),
  2014}.\hskip 1em plus 0.5em minus 0.4em\relax IEEE, 2014, pp. 3579--3584.

\bibitem{khanafer2014stability2}
A.~Khanafer, T.~Ba{\c{s}}ar, and B.~Gharesifard, ``Stability properties of
  infection diffusion dynamics over directed networks,'' in \emph{Decision and
  Control (CDC), 2014 IEEE 53rd Annual Conference on}.\hskip 1em plus 0.5em
  minus 0.4em\relax IEEE, 2014, pp. 6215--6220.

\bibitem{bonaccorsi2015epidemic}
S.~Bonaccorsi, S.~Ottaviano, D.~Mugnolo, and F.~D. Pellegrini, ``Epidemic
  outbreaks in networks with equitable or almost-equitable partitions,''
  \emph{SIAM Journal on Applied Mathematics}, vol.~75, no.~6, pp. 2421--2443,
  2015.

\bibitem{van2009performance}
P.~Van~Mieghem, \emph{Performance analysis of communications networks and
  systems}.\hskip 1em plus 0.5em minus 0.4em\relax Cambridge University Press,
  2009.

\bibitem{rudin1964principles}
W.~Rudin, \emph{Principles of mathematical analysis}.\hskip 1em plus 0.5em
  minus 0.4em\relax McGraw-Hill New York, 1964, vol.~3.

\bibitem{Girvan2002}
M.~Girvan and M.~E. Newman, ``Community structure in social and biological
  networks,'' \emph{Proceedings of the national academy of sciences}, vol.~99,
  no.~12, pp. 7821--7826, 2002.

\bibitem{slide}
S.~Sternberg, ``The perron-frobenius theorem.''
  \url{http://www.math.harvard.edu/library/sternberg/slides/1180912pf.pdf}.

\end{thebibliography}
\bibliographystyle{IEEEtran}

%







\newpage
\setcounter{page}{1}

\appendix
\renewcommand{\theequation}{A.\arabic{equation}}
\setcounter{equation}{0}

\subsection*{Proof of Lemma \ref{Lemma: Monotonicity}}

\begin{proof}
Since $F$ is concave in $\mathbb{R}_+^n$,
\[F(\frac{x+y}{2})\succeq \frac{1}{2}F(x)+\frac{1}{2}F(y).\]
Furthermore, since $F$ is homogeneous, $F(x+y)\succeq F(x)+F(y)$, i.e., $F$ is super-additive.

\noindent For $0\preceq x\preceq y$,
\[F(y)=F(x+(y-x))\succeq F(x)+F(y-x) \succeq F(x).\]
{The last inequality is due to having $F(y-x)\succeq 0$} for $y\succeq x$. Therefore, $F$ is monotone.
\end{proof}

\subsection*{Proof of Lemma \ref{Lemma: C1C2}}
\begin{proof}
Since $F$ is homogeneous and concave, $F(e_J)=F(\sum_{j\in J}e_j)\succeq \sum_{j\in J}F(e_j)$. Therefore, if $F$ satisfies condition {\bf C1}, there exists exists $j\in J$ and $i\notin J$ such that $F_i(e_j)>0$, so it is also true that $F_i(e_J)>0$. Hence, $F$ must also satisfy {\bf C2}.
\end{proof}

\subsection*{Proof of Lemma \ref{Lemma: allx}}

\begin{proof}
For $x\succ 0$, let $x_{min}>0$ be the smallest entry of $x$, and $u$ be the vector of ones. Then, $w\triangleq x-x_{min}u\succeq 0$. So, replacing for $x=w+x_{min}u$ in $F_i(x\circ e_J)$ yields
\[
F_i(x\circ e_J)=F_i(w\circ e_J+x_{min}e_J)\geq F_i(w\circ e_J)+x_{min}F_i(e_J)>0
\]

\end{proof}

\subsection*{Proof of Theorem \ref{Th: Primitive}}

\begin{proof}
Let $y_m\triangleq F^m_c(x)$, i.e., $y_0=x$, $y_1=F_c(x)$, and so on. With this definition, we can write a recursive formula for $y_m$ with $m>0$.
\[y_m=F_c(y_{m-1})=cy_{m-1}+F(y_{m-1}).\]

Now, let $J_m$ be the index set of positive elements of $y_m$, i.e., 
\[J_m\triangleq\{j|(y_m)_j>0\}.\]
Since $x\succnsim 0$, $J_0=\{j|x_j>0\}$ is non-empty. Also, $J_0\neq\{1,2,...,n\}$ unless $x\succ 0$ for which $F^m_c(x)\succ 0$ and no further investigation is required.

For any $\emptyset\neq J_m\subsetneq \{1,...,n\}$, define $I_m$ as
\[I_m\triangleq\{i|F_i(y_m)>0, i\notin J_m\}.\]
Since $F$ satisfies condition {\bf C2}, $I_m$ is not empty.

According to the recursive formula for $y_m$, the positive elements of $y_m$ either are also positive in $y_{m-1}$, or they correspond to the positive elements of $F(y_{m-1})$. Expressed in terms of $I$ and $J$ sets,
\[J_m=J_{m-1}\cup I_{m-1}.\]
Since by definition $J_{m-1}$ and $I_{m-1}$ do not intersect, and $I_{m-1}$ is not empty as long as $|J_{m-1}|\neq n$, the cardinality of $J_m$ strictly increases by $m$. Therefore, at some iteration step $M\leq n-|J_0|\leq n-1$, we get $|I_M|=n$. In other words, for some $M$, all the elements of $F^M_c(x)$ are positive, i.e., $F^M_c(x)\succ 0$.

Now for the necessity part, we show that if $F$ does not satisfy condition {\bf C2}, then $F_c$ is not primitive. Let $J$ be a non-empty set such that $F_i(e_J)=0$ for all $i\notin J$. According to Lemma \ref{Lemma: allx}, $F_i(x\circ e_J)=0$ for any $x\succ 0$. As a result, $F_c(x\circ e_J)=x'\circ e_J$ for some $x'\succ 0$. Repeating this process yields $F^2_c(x\circ e_J)=F_c(x'\circ e_J)=x''\circ e_J$, and so on. Therefore, for any iterate $m>0$, the $i-$th entry of $F^m_c(x\circ e_J)$ is equal to zero.

\end{proof}

\subsection*{Proof of Theorem \ref{peron} (Nonlinear Perron--Frobenius)}\label{Sec: NPFproof}

To prove the existence, uniqueness, and strict positivity of the solution to nonlinear Perron--Frobenius problem $F(x)=\lambda x$, we use the following lemma.

Recall that if $F$ satisfies the condition {\bf C2}, there exists an $M$ such that $F^M_c(x)\succ 0$ for all $x\succnsim 0$. Let define $P(x)\triangleq F^M_c(x)$ for one such $M$.

\begin{lemma}\label{Lemma: P-Monotonicity}
For the homogeneous, concave map $F$ of the non-negative cone that satisfies the condition {\bf C2}, following statements are true for $P(x)$:
\begin{enumerate}
\item $P(x)$ is homogeneous, concave, and super-additive.
\item For all $x\succeq 0$,
    \begin{equation}\label{PF-FP}
	P(F(x))\preceq F(P(x)).
	\end{equation}
\item For $0\preceq x\precnsim y$,
    \begin{equation}\label{p}
	P(x)\prec P(y).
	\end{equation}  
\end{enumerate}
  
\end{lemma}

\begin{proof}
\begin{enumerate}
\item Since, $F$ is homogeneous and concave, so is $F_c(x)=cx+F(x)$ with $c>0$. Same is true for $F_c^m(x)$. Moreover, $F_c^m(x)$ is super-additive according to Lemma \ref{Lemma: Monotonicity}.
\item 
{\footnotesize
\footnotesize
\begin{multline*}
F(P(x))=F(F^M_c(x))=F_c(F^M_c(x))-cF^M_c(x)\\
=F^M_c(F_c(x))-cF^M_c(x)=F^M_c(F(x)+cx)-cF^M_c(x)\\
\succeq F^M_c(F(x))+cF^M_c(x)-cF^M_c(x)=F^M_c(F(x))=P(F(x)).
\end{multline*}
}
\item For $0\preceq x\precnsim y$,
\[P(y)=P(x+(y-x))\succeq P(x)+P(y-x) \succ P(x),\]
because $P(y-x)\succ 0$ for $(y-x)\succnsim 0$ when $F$ satisfies condition {\bf C2}.
\end{enumerate}

\end{proof}
   
For proving convergence part of Theorem \ref{peron}, we use the following lemma and a proposition.

\begin{proposition}[Theorem 9 in \cite{krause2001concave}]
	If $F:\mathbb{R}_{+}^{n}\rightarrow\mathbb{R}_{+}^{n}$ is a homogeneous, primitive, and concave map of nonnegative cone:
	\[\lim\limits_{k\rightarrow\infty} \frac{F^k(x)}{||F^k(x)||}=x^*,\text{ for all }x\succnsim 0.\]
    where $x^*\succ 0$ is the eigenvector of $F$. \label{Proposition: Convergence}
\end{proposition}

Now, with these background tools we prove Theorem \ref{peron}.
	\begin{proof}
		
		\textbf{Existence:} Here, we follow some of ideas in \cite{slide}.
		Assume for all $x\in \mathbb{R}_{+}^{n}$
		\begin{equation}
		\rho(x)=\max\lbrace m\vert mx\preceq F(x)\rbrace.
        \label{Rhox}
		\end{equation}
		Because the set $S=\lbrace x\vert x\in \mathbb{R}_{+}^{n},~\Vert x\Vert=1 \rbrace$  is a compact set, $\rho(x)$ has a maximum value over $S$. Let $x_{0}\in S$ be a vector that maximizes $\rho(x)$. 
        
       Let us assume that $\rho(x_{0})x_{0} \neq F(x_{0})$. If so, then $\rho(x_{0})x_{0}\precnsim F(x_{0})$, and thus, Eq. (\ref{p}) of Lemma \ref{Lemma: P-Monotonicity} yields
		\[
		P(\rho(x_{0})x_{0}) \prec P(F(x_{0})). 
		\] 
		Since $P$ is homogeneous as the result of $F$ homogeneity, the left hand side of the the above inequality becomes $\rho(x_{0})P(x_{0})$. Furthermore, according to Eq. (\ref{PF-FP}) of Lemma \ref{Lemma: P-Monotonicity}, the right hand side preceeds $F(P(x_{0}))$. Therefore, we have
		\[
		\rho(x_{0})P(x_{0}) \prec F(P(x_{0})). 
		\] 
		Define $x_1\triangleq\frac{P(x_0)}{\Vert P(x_0)\Vert}\in S$. The above inequality becomes $\rho(x_{0})x_{1} \prec F(x_{1})$, which according to the definition of $\rho(x)$ in Eq. (\ref{Rhox}) suggests that $\rho(x_1)>\rho(x_{0})$. However, this is contradicting the fact that $\rho(x_{0})$ is a maximum value. Hence, by contradiction, we arrive at the conclusion that 
		\[
		\rho(x_{0})x_{0} =F(x_{0}), 
		\]
		showing the existence of an eigenvector in $\mathbb{R}_{+}^{n}$.
		
		\textbf{Strict positivity:} So far, we proved that a non-negative eigenvector exists for the nonlinear map $F$ with a positive eigenvalue. We can show that any non-negative eigenvector $y$ is indeed strictly positive. If $0\precnsim y$ is an eigenvector of $F$, then its corresponding eigenvalue must be non-negative ($\lambda\geq0$) because $F(y)=\lambda y$ is non-negative. Furthermore, for $0\precnsim y$, Eq. (\ref{p}) yields
		\[
		0\prec P(y)=F_c^M(y)=(\lambda+c)^My,
		\] 
Since $c>0$, the coefficient $\eta(\lambda)=(\lambda+c)^M$ is also positive. Thus, $y$ must be strictly positive.
		
		\textbf{Uniqueness}: Now we show all the eigenvectors of $F$ in nonnegative cone are a scalar multiplication of $x_{0}$. Assume $x_{1}\succ 0$ is an eigenvector of $F$, that is $F(x_{1})=\lambda x_{1}$. Moreover from the definition of $\rho(x_{1})$ we have 
		\[
		\lambda=\rho(x_{1})\leq \rho(x_{0}).
		\]
        Let $r=\max_i \frac{(x_0)_i}{(x_1)_i}$. Assume that $x_1\succ 0$ is not a multiple of $x_0\succ 0$. Therefore, the vector $z=rx_1-x_0\succnsim 0$. Eq. (\ref{p}) yields
        \[
		\begin{split}
		0\prec P(z)&=P(rx_1-x_0)\\
        &\preceq rP(x_{1})-P(x_{0})\\
		&=r\eta(\lambda)x_{1}-\eta(\rho(x_{0}))x_{0}\\
		&\preceq \eta(\rho(x_{0}))(rx_{1}-x_{0}).
		\end{split}
		\]
  Which concludes that $z\succ 0$. But this is not possible because at least one element of $z$ is zero.      
        Therefore, by contradiction $z=\max_i \frac{(x_0)_i}{(x_1)_i}x_1-x_0=0$, proving that all the eigenvectors of $F$ in the nonnegative cone are a scalar multiplication of $x_{0}$.

\textbf{Necessity of C2:} So far, we proved that {\bf C2} is a sufficient condition for existence, strict positivity, and uniqueness of an eigenvector for $F$. Here we prove that {\bf C2} is also a necessary condition.

The approach is to show that if $F$ does not satisfy {\bf C2}, then $F$ has another eigenvector $y^*\in \mathbb{R}_+^N$ which is not a scalar multiple of $x^*$.

Let first define for $J=\{j_1,j_2,...,j_m\}$, the projection $\pi_J: \mathbb{R}_+^N\rightarrow \mathbb{R}_+^{|J|}$ so that for any $x\in\mathbb{R}_+^N$, $(\pi_J(x))_k=x_{j_k}$ for $k=1,...,|J|$. Similarly, let define the inverse projection $\pi_J^{-1}: \mathbb{R}_+^{|J|}\rightarrow \mathbb{R}_+^N$ so that for any $z\in\mathbb{R}_+^{|J|}$, the vector $x=\pi_J^{-1}(z)$ has zeros on entries not belonging to $J$ (i.e., $x_j=0$ for all $j\notin J$), and for those entries belonging to $J$, $x_{j_k}=z_k$.

Assume $F$ does not satisfy {\bf C2}. Therefore, there exists $J\subset\{1,...,N\}$ so that $F_i(e_J)=0$ for all $i\notin J$. Define $\mathcal{J}$ as the set of all such $J$ sets. Among all the members of $\mathcal{J}$ pick a set $J_*$ so that no other $J\in\mathcal{J}$ is a subset of $J_*$. The key characteristic of the subset $J_*$ is that the function $H(z)\triangleq\pi_{J_*}( F ( \pi^{-1}_{J_*}(z)))$ satisfies condition {\bf C2} in  $\mathbb{R}_+^{|J_*|}$, otherwise a subset of $J_*$ would belong to $\mathcal{J}$ which is not possible by the definition of $J_*$.

Since $H$ satisfies the condition {\bf C2} in  $\mathbb{R}_+^{|J_*|}$, the nonlinear Perron-Frobenius problem $H(z)=\lambda z$ has a unique solution $z^*\succ 0$. This yields that $\pi^{-1}(z^*)\in \mathbb{R}_+^N$ must also be an eigenvector of $F$. Note that since $\pi^{-1}(z^*)$ has zeros on some of its entries it cannot be a scalar multiple of $x^*$. This violates uniqueness of eigenvectors of $F$ in $\mathbb{R}_+^N$. Therefore, through proof by contradiction we conclude that $F$ must satisfy {\bf C2}.

		\textbf{Convergence:} As proved in Theorem \ref{Th: Primitive}, $F_c(x)=F(x)+cx$ is primitive. Therefore, Proposition \ref{Proposition: Convergence} is applicable.
\end{proof}

\subsection*{Proof of Theorem \ref{Th: MC2}}

\begin{proof}
	For the proof, we leverage Theorem \ref{Th: Primitive} in that $F$ satisfying condition \textbf{C2} is equivalent to $F_c(x)=F(x)+cx$ being primitive.
    
	If we show that for any $j$ there exist an $M_{j}$ such that $F_c^{M_j}(e_j)\succ 0$ then we can choose $M$ as ${\displaystyle\max_{j}\{M_{j}\}}$ and from super-additivity of $F$ we can conclude that for any non-empty set $J$, $F_c^M(e_J)\succ 0$, i.e., $F_c$ is primitive.

We employ the definition of graphs $\boldsymbol{G}^k$ associated with the multilayer network $\mathcal{G}=(V,E_S,E_A)$ as explained in Section \ref{Sec: Multilayer}. According to Definition \ref{def: sc}, if $\mathcal{G}$ is M--connected, there exists $k_*$ such that $\boldsymbol{G}^{k_*}$ is a strongly connected graph.

Therefore, choosing $i$ from $G^{1}$ such that $(i,j)\in \mathcal{L}_1$ (for which $(i,j)\in E_A$ and  $(i,j)\in E_B$ necessitates $F_i(e_j)>0,$ according to the statement of the theorem) yields $y_1=F_c(e_j)$ has positive values on entries $i$ and $j$. Therefore, with finite iterations of $F_c$, say $m_1$ times, $y_{m_1}=F_c^{m_1}(e_j)$ will have positive values on all the nodes in the strongly connected component of $G^{1}$ that contains $j$. This strongly connected component of $G^{1}$ becomes a single hypernode, which we call $J\in \mathcal{P}_2$, for graph $\boldsymbol{G}^2$. According to the definition of $\boldsymbol{G}^k$, for graph $\boldsymbol{G}^2$, there is a directed link from component $J$ to component $I$, i.e., $(I,J)\in\mathcal{L}_2$, if and only if 
there exists $i\in I$ and $j_{1},j_{2}\in J$ such that $(i,j_{1})\in E_A$, and $(i,j_2)\in E_B$. Since $\forall j\in J,\, (y_{m_1})_j>0$, we get $F_i(y_{m_1})>0$ for the above choice of $i$, which further indicates all the elements of $I$ will have positive in $y_{m_2}=F_c^{m_2}(e_j)$ values after finite iterations. As a result, all the nodes belonging to the strongly connected component of $\boldsymbol{G}^2$ that contains $J$ will have positive entries. This procedure can be repeated for $\boldsymbol{G}^3,...,\boldsymbol{G}^{k_*}$. Since, $\boldsymbol{G}^{k_*}$ is strongly connected, all the nodes of the network must have positive values in $F_c^{M_j}(e_J)$.

	
	Now, assume $\mathcal{G}$ is not M--connected. Thus, the iteration of $k$ should not change $\boldsymbol{G}^k$ after some finite steps. This graph, which we refer to as $\boldsymbol{G}^\infty$, is a \textit{directed acyclic graph} (DAG) because it cannot have any strongly connected components. Let $J$ denote the source hypernode of $\boldsymbol{G}^\infty$ (which always exists for a finite DAG). Therefore, there is no $j_1,j_2\in J$ and $i \notin J$ for which $(i,j_1)\in E_A$ and $(i,j_2)\in E_B$. For this set $J$, the elements of nodes not belonging to $J$ will always have zero values in $F_c^{m}(e_J)$ for any $m$. Therefore, it is not possible to find an $M$ such that $F_c^{m}(e_J)\succ 0$.
\end{proof}

\subsection*{Proof of Theorem \ref{TheoremAprox} (perturbation formulas):}
\begin{proof}
For small values of $\bar{\kappa}$, we can find the solution to the threshold
equation (\ref{ThresholdEquation}), by perturbation at $\bar{\kappa}=0$.
Taking right derivative of threshold equation (\ref{ThresholdEquation} ) with
respect to $\bar{\kappa}$ at $\bar{\kappa}=0$ and $\boldsymbol{z}=\boldsymbol{v}$, yields%
\begin{align}
\frac{dz_{i}}{d\bar{\kappa}}  &  =\frac{d\tau_{c}}{d\bar{\kappa}}%
\sum\nolimits_{j=1}^{N}w_{ij}^{S}v_{j}\nonumber\\
&  +\frac{1}{\lambda_{1}(W_S)}\sum\nolimits_{j=1}^{N}w_{ij}^{S}\frac{dz_{j}%
}{d\bar{\kappa}}\nonumber\\
&  +\frac{1}{\lambda_{1}(W_S)}\{1-\frac{\sum w_{ij}^{S}v_{j}}{\sum w_{ij}%
^{A}v_{j}}\}\sum\nolimits_{j=1}^{N}w_{ij}^{S}v_{j}\nonumber\\
&  =\lambda_{1}(W_S)\frac{d\tau_{c}}{d\bar{\kappa}}-\frac{\lambda_{1}%
(W_S)v_{i}^{2}}{\sum w_{ij}^{A}v_{j}}+v_{i}\nonumber\\
&  +\frac{1}{\lambda_{1}(W_S)}\sum\nolimits_{j=1}^{N}w_{ij}^{S}\frac{dz_{j}%
}{d\bar{\kappa}} \label{dwi}%
\end{align}

The collective form of (\ref{dwi}) is%
\begin{equation}
\frac{d\boldsymbol{z}}{d\bar{\kappa}}=\lambda_{1}(W_{S})\frac{d\tau_{c}}%
{d\bar{\kappa}}\boldsymbol{v}+(I-D_{0})\boldsymbol{v}+\frac{1}{\lambda
_{1}(W_{S})}W_{S}\frac{d\boldsymbol{z}}{d\bar{\kappa}},
\end{equation}
where $D_{0}$ is the diagonal matrix%
\begin{equation}
D_{0}\triangleq diag\{\frac{\sum\nolimits_{j=1}^{N}w_{ij}^{S}v_{j}}%
{\sum\nolimits_{j=1}^{N}w_{ij}^{A}v_{j}}\}.
\end{equation}
Therefore, following is also true%
\begin{equation}
(I-\frac{1}{\lambda_{1}(W_{S})}W_{S})\frac{d\boldsymbol{z}}{d\bar{\kappa}%
}=\lambda_{1}(W_{S})\frac{d\tau_{c}}{d\bar{\kappa}}\boldsymbol{v}%
-D_{0}\boldsymbol{v}+\boldsymbol{v}.
\end{equation}
Multiplying both sides by $\boldsymbol{u}^{T}$ from left gives%
\begin{equation}
0=\lambda_{1}(W_{S})\frac{d\tau_{c}}{d\bar{\kappa}}-\boldsymbol{u}^{T}%
D_{0}\boldsymbol{v}+1, \label{dtawc1}%
\end{equation}
because $\boldsymbol{u}$\ is the left dominant eigenvector of $W_{S}$ and is
normalized such that $\boldsymbol{u}^{T}\boldsymbol{v}=1$. From (\ref{dtawc1}%
), we get
\[
\frac{d\tau_{c}}{d\bar{\kappa}}=\frac{1}{\lambda_{1}(W_{S})}(\Psi(W_{S}%
,W_{A})-1),
\]
because $\boldsymbol{u}^{T}D_{0}\boldsymbol{v}=\Psi(W_{S},W_{A})$ according to
(\ref{PHIAB}). Formula (\ref{tawc1}) is the first order Taylor expansion of
$\tau_{c}$ at $\bar{\kappa}=0$.

For large values of $\bar{\kappa}$, we can find the solution to the threshold
equations (\ref{tawc1}\&\ref{tawc2}), by perturbation at $s=0$, where
$s\triangleq\bar{\kappa}^{-1}$. The threshold equation
(\ref{ThresholdEquation}) in term of the new parameter $s$ is%
\begin{align}
z_{i}  &  =\tau_{c}\{\frac{\bar{\kappa}^{-1}}{\bar{\kappa}^{-1}}\frac
{(\bar{\kappa}+1)\sum w_{ij}^{A}z_{j}}{\bar{\kappa}\sum w_{ij}^{S}z_{j}+\sum
w_{ij}^{A}z_{j}}\}\nonumber\\
&  =\tau_{c}\{\frac{(1+s)\sum w_{ij}^{A}z_{j}}{\sum w_{ij}^{S}z_{j}+s\sum
w_{ij}^{A}z_{j}}\}\sum w_{ij}^{S}z_{j}\nonumber\\
&  =\tau_{c}\{\frac{(s+1)\sum w_{ij}^{S}z_{j}}{s\sum w_{ij}^{A}z_{j}+\sum
w_{ij}^{S}z_{j}}\}\sum w_{ij}^{A}z_{j}. \label{wiB}%
\end{align}

As can be seen, (\ref{wiB}) has exactly the form of the threshold equation (\ref{ThresholdEquation}), where $W_{S}$ and $W_{A}$ matrices have changed
roles and $\bar{\kappa}$ is replaced by $s$. Therefore, similar to the proof
of Theorem \ref{Theorem: NPF}, the threshold can be found around $s=0$ as by
switching $W_{S}$ and $W_{A}$ matrices and replacing $\bar{\kappa}$ by
$s=\bar{\kappa}^{-1}$. The result is Eq. (\ref{tawc2}), which is in fact the first order Taylor expansion of $\tau_{c}$ at $s=0$.
\end{proof}

\subsection*{Proof of Lemma \ref{Theorem: Adverse} (lower-bound for symmetric layers): \label{append: symmetric}}
\begin{proof}
For $A$ and $B$ symmetric, we have%

\begin{align*}
\Psi(A,B)  &  =\sum\nolimits_{i=1}^{N}u_{i}v_{i}\frac{
{\textstyle\sum\limits_{j=1}^{N}}
a_{ij}v_{j}}{
{\textstyle\sum\limits_{j=1}^{N}}
b_{ij}v_{j}}=\lambda_{1}(A)\sum\nolimits_{i=1}^{N}\frac{v_{i}^{3}}{
{\textstyle\sum\limits_{j=1}^{N}}
b_{ij}v_{j}}\\
&  =\lambda_{1}(A)\frac{
{\textstyle\sum\limits_{i=1}^{N}}
\left(  \frac{v_{i}^{3}}{
{\textstyle\sum\limits_{j=1}^{N}}
b_{ij}v_{j}}\right)  ^{2/2}
{\textstyle\sum\limits_{i=1}^{N}}
\left(  v_{i}
{\textstyle\sum\limits_{j=1}^{N}}
b_{ij}v_{j}\right)  ^{2/2}}{\sum\nolimits_{i=1}^{N}v_{i}
{\textstyle\sum\nolimits_{j=1}^{N}}
b_{ij}v_{j}}\\
&  \geq\lambda_{1}(A)\frac{\left(
{\textstyle\sum\limits_{i=1}^{N}}
\left(  \frac{v_{i}^{3}}{
{\textstyle\sum\limits_{j=1}^{N}}
b_{ij}v_{j}}\right)  ^{1/2}\left(  v_{i}
{\textstyle\sum\limits_{j=1}^{N}}
b_{ij}v_{j}\right)  ^{1/2}\right)  ^{2}}{\sum\nolimits_{i=1}^{N}v_{i}
{\textstyle\sum\nolimits_{j=1}^{N}}
b_{ij}v_{j}}\\
&  =\lambda_{1}(A)\frac{\left(  \sum\nolimits_{i=1}^{N}v_{i}^{2}\right)  ^{2}%
}{\sum\nolimits_{i=1}^{N}v_{i}
{\textstyle\sum\nolimits_{j=1}^{N}}
b_{ij}v_{j}}=\lambda_{1}(A)\frac{1}{v^{T}Bv}\\
&  \geq\frac{\lambda_{1}(A)}{\max_{\left\Vert x\right\Vert =1}x^{T}Bx}%
=\frac{\lambda_{1}(A)}{\lambda_{1}(B)},
\end{align*}
where the first inequality is due to the Cauchy--Schwarz inequality and the
last one is according to the Rayleigh quotient definition of the largest eigenvalue.
\end{proof}

\end{document}